\documentclass{svjour3}

\usepackage[utf8]{inputenc}
\usepackage[T1]{fontenc}
\usepackage{lmodern}
\usepackage{graphicx}
\usepackage{amssymb}
\usepackage{amsmath}
\usepackage{subcaption}
\usepackage{algorithm}
\usepackage{algorithmic}
\usepackage{color}

\journalname{Machine learning}

\begin{document}
\raggedbottom

\title{Boolean matrix logic programming for active learning of gene functions in genome-scale metabolic network models}

\author{Lun Ai \and Stephen H. Muggleton \and Shi-Shun Liang \and Geoff S. Baldwin}


\institute{Lun Ai (corresponding author) \at
           Department of Computing, Imperial College London, London, UK\\
           \email{lun.ai15@imperial.ac.uk}
           \and
           Stephen H. Muggleton \at
           Department of Computing, Imperial College London, London, UK\\
           \email{s.muggleton@imperial.ac.uk}
           \and
           Shi-Shun Liang \at
           Department of Life Sciences, Imperial College London, London, UK\\
           \email{shishun.liang20@imperial.ac.uk}
           \and
           Geoff S. Baldwin\at
           Department of Life Sciences, Imperial College London, London, UK\\
           \email{g.baldwin@imperial.ac.uk}
}
\date{}

\maketitle
\begin{abstract}
\vspace{-50pt}
Reasoning about hypotheses and updating knowledge through empirical observations are central to scientific discovery. In this work, we applied logic-based machine learning methods to drive biological discovery by guiding experimentation. Genome-scale metabolic network models (GEMs) - comprehensive representations of metabolic genes and reactions - are widely used to evaluate genetic engineering of biological systems. However, GEMs often fail to accurately predict the behaviour of genetically engineered cells, primarily due to incomplete annotations of gene interactions.  The task of learning the intricate genetic interactions within GEMs presents computational and empirical challenges. To efficiently predict using GEM, we describe a novel approach called Boolean Matrix Logic Programming (BMLP) by leveraging Boolean matrices to evaluate large logic programs. We developed a new system, $BMLP_{active}$, which guides cost-effective experimentation and uses interpretable logic programs to encode a state-of-the-art GEM of a model bacterial organism. Notably, $BMLP_{active}$  successfully learned the interaction between a gene pair with fewer training examples than random experimentation, overcoming the increase in experimental design space. $BMLP_{active}$ enables rapid optimisation of metabolic models to reliably engineer biological systems for producing useful compounds. It offers a realistic approach to creating a self-driving lab for biological discovery, which would then facilitate microbial engineering for practical applications.

\end{abstract}

\section{Introduction}
\label{sec:introduction} 

In the pursuit of understanding natural phenomena, scientists often formulate plausible theories and design experiments to test competing hypotheses. This process typically involves careful experimental design and abductive reasoning. Experimentation can be optimised through active learning \cite{mitchell_generalization_1982,cohn_improving_1994}, where informative training examples are strategically selected to combat resource limitations. Inductive Logic Programming (ILP) \cite{ILP1991} offers an automated approach to abduction, representing observations, hypotheses and background knowledge through interpretable logic programs. In ILP, hypotheses are learned to explain observational data with respect to the background knowledge. ``Askable'' hypotheses, known as abducibles, that might explain contradictions between observations and prior background knowledge are proposed \cite{ase_progol,TCIE}. The extended knowledge, enriched with abduced hypotheses, is subsequently verified against observations. 

Our work explored the applicability of abductive reasoning and active learning to identify the functions of metabolic genes in model bacterial organisms. Given that biological relationships are commonly described logically, ILP is particularly adept at operating on biological knowledge bases. The integration of abductive reasoning and active learning via ILP was successfully demonstrated in the context of biological discovery by the Robot Scientist \cite{King04:RobotScientist}. The Robot Scientist performed active learning by strategically selecting key experiments to achieve more data- and cost-effective gene function learning than random experimentation. However, this demonstration was limited to only 17 genes in the aromatic amino acid pathway of yeast. Our work significantly advanced the application of this integrative learning paradigm to genome-scale biological discoveries by examining genome-scale metabolic models (GEMs). We examined the GEM model iML1515 \cite{iML1515}, which encompasses 1515 genes and 2719 metabolic reactions of the \textit{Escherichia coli} (\textit{E. coli}) strain K-12 MG1655, a versatile organism for metabolic engineering to produce specific compounds. 

\begin{figure}[t]
    \centering
        \includegraphics[width=\linewidth]{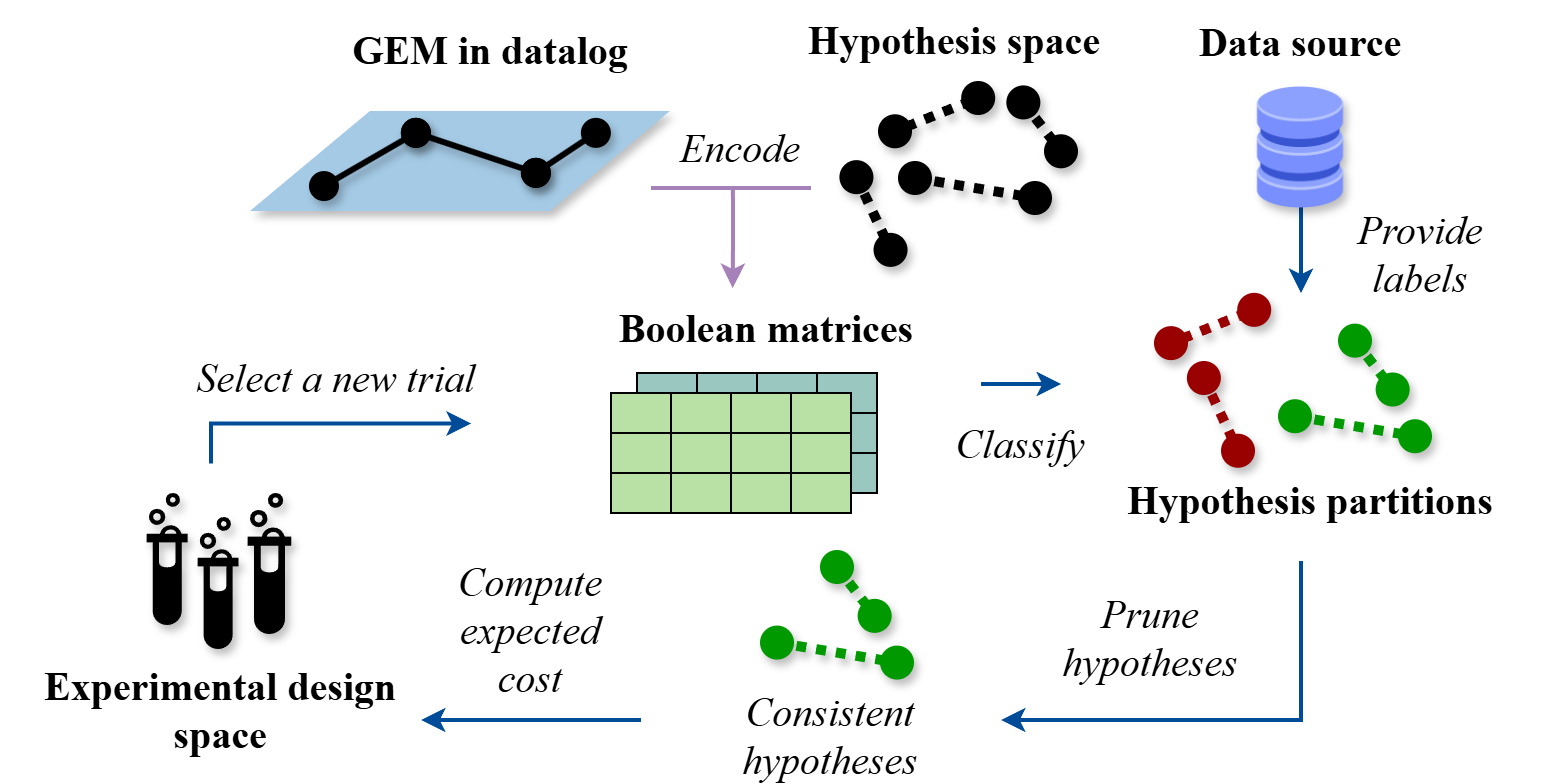} 
    \caption{$BMLP_{active}$ encodes the GEM iML1515 as Boolean matrices and predicts the auxotrophic mutant phenotypes. It actively consults a data source to request ground truth labels, which minimises the expected experimental cost based on a user-provided cost function. The underlying BMLP iteratively refutes gene function annotation hypotheses inconsistent with labelled training examples. While $BMLP_{active}$ to date only learns from synthetic data, a laboratory robot can be integrated to perform selected experiments.}
    \label{fig:framework_overview}
\end{figure}

Despite the extensive mapping of the \textit{E. coli} genome to metabolic functions in iML1515, inaccurate phenotype predictions have been identified due to erroneous gene function annotations \cite{bernstein_evaluating_2023}. Learning gene function annotations in GEMs is computationally and empirically challenging. We built the first logic programming system on GEMs $BMLP_{active}$ (Fig. \ref{fig:framework_overview}) based on our novel Boolean Matrix Logic Programming (BMLP) approach. $BMLP_{active}$ uses Boolean matrices to encode the biochemical and genetic relationships in iML1515 and enables high-throughput logical inferences to predict auxotrophic mutant phenotypes. $BMLP_{active}$ implements active learning to select auxotrophic mutant experiments that minimise the expected experimental cost via a user-defined cost function. $BMLP_{active}$ successfully learned gene function annotations through active learning while reducing 90\% of the optional nutrient substance cost required by randomly selected experiments. It additionally reduced the number of experimental data needed compared to random experiment selection. When operating within a finite budget, our approach could deliver optimal experimental outcomes, whereas studies employing random experiment selection might not reach completion. 

Furthermore, we applied $BMLP_{active}$ to learning digenic function annotations, which holds significant implications for drug development and therapeutic interventions \cite{costanzo_global_2016}. Digenic functions are related to isoenzymes - two enzymes responsible for the same reaction. Digenic interactions between complex genotypes and phenotypes remains largely unexplored in most organisms  \cite{costanzo_global_2019}, and their dynamics depend on the growth conditions \cite{bernstein_evaluating_2023}. A comprehensive understanding of digenic interactions would expedite strain engineering efforts. We show that $BMLP_{active}$ converged to the correct gene-isoenzyme mapping with as few as 20 training examples. This represents a significant promise of $BMLP_{active}$ to address complex genetic interactions on the whole genome level.

\section{Representing a GEM using datalog}

We refer the readers to \cite{LP_book,ILP_foundation} for terminology on datalog programs. A recursive program has a body literal that appears in the head of a clause. Specifically, we focused on linear and immediately recursive datalog as simple recursive datalog programs. A linear and immediately recursive datalog program has a single recursive clause where the recursive predicate appears in the body only once \cite{ioannidis_computation_1986}. This subset of recursive datalog programs can represent the relational structure of reaction pathways:
\begin{flalign*}
&pathway(X,Y) \gets reaction(X,Y).  \\
&pathway(X,Y) \gets reaction(X,Z), pathway(Z,Y).
\end{flalign*}
Petri nets are a class of directed bipartite graphs \cite{reisig_understanding_2013}, and they are a commonly used graphical tool to model processes in biological systems \cite{sahu_advances_2021}. we refer readers to \cite{reisig_understanding_2013,rozenberg_elementary_1998} for background on Petri net. Petri nets contain nodes marked by tokens (black dots), which indicate the availability of resources such as chemical metabolites. The presence of reactants or products determines which reaction pathways are viable and how the resources are allocated. We examined a class of Petri nets called the one-bounded elementary net (OEN) \cite{reisig_understanding_2013}, which allocates at most one token per node (definitions in Appendix \ref{app:OEN}). We use OEN as a conceptual model of essential metabolites and pathways in a metabolic network. We provided an example of a simple metabolic network encoded as a linear and immediately recursive datalog program $\mathcal{P}_1$. 

\begin{minipage}{0.4\textwidth}
\includegraphics[width=0.9\linewidth]{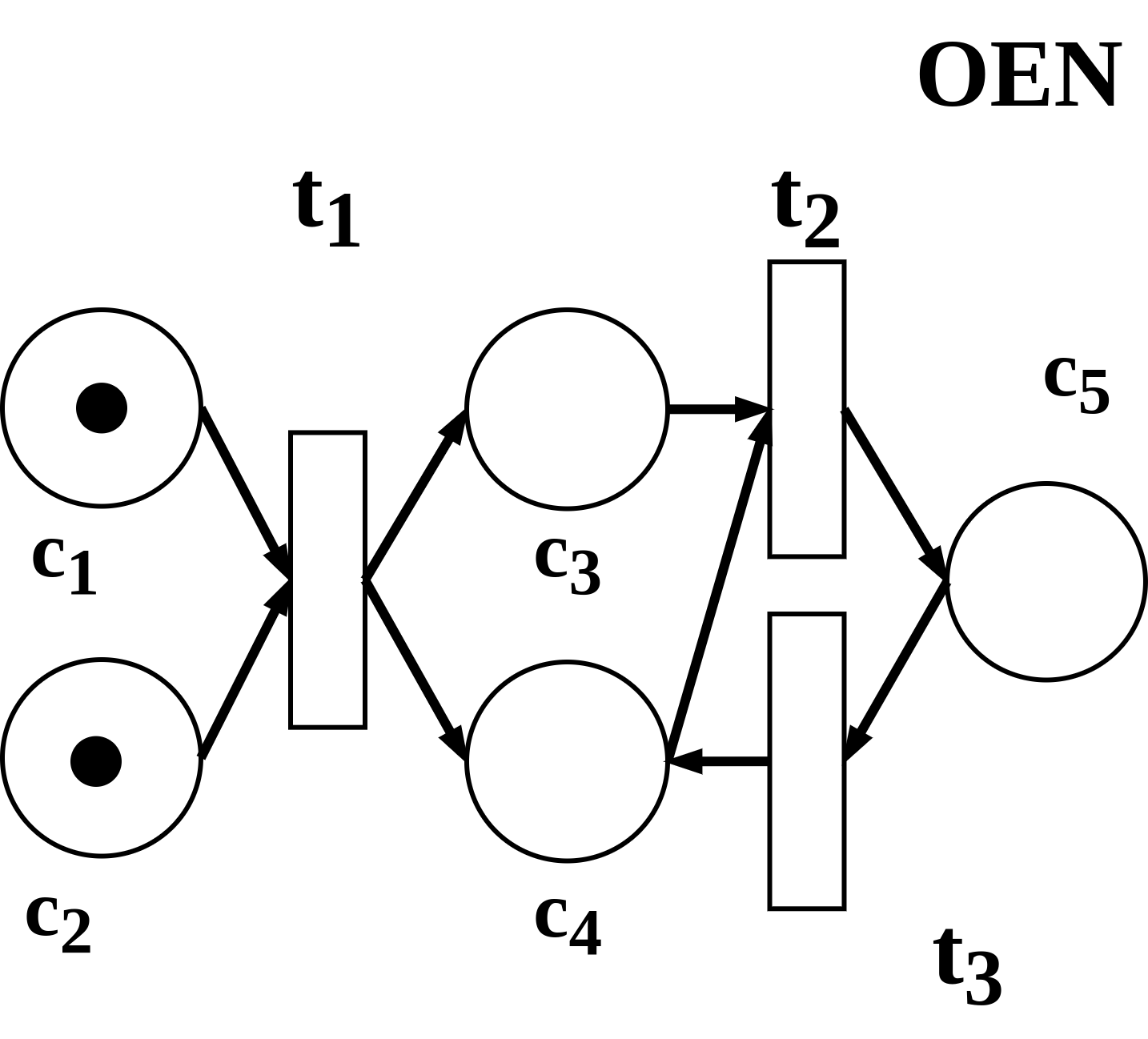}
\end{minipage}
\begin{minipage}{0.57\textwidth}
\begin{gather}
\mathcal{P}_{1}:\left\{
\begin{aligned}
    reaction\_1(c1\_c2, c3\_c4). &\\
    reaction\_1(c1\_c2, c3). &\\
    reaction\_1(c1\_c2, c4).&\\
    reaction\_2(c3\_c4, c5). &\\
    reaction\_3(c5, c4). &\\
    reaction(X,Y) \gets reaction\_1(X,Y). &\\
    reaction(X,Y) \gets reaction\_2(X,Y). &\\
    reaction(X,Y) \gets reaction\_3(X,Y). &\\
    pathway(X,Y) \gets reaction(X,Y). &\\ 
    pathway(X,Y) \gets reaction(X,Z), &\\
                pathway(Z,Y). & \nonumber
\end{aligned}
\right\}
\end{gather}
\end{minipage}

Each ground fact in the program describes a reaction between a set of reactants and a set of products. Evaluating the recursive program $pathway(X, Y)$ is equivalent to finding the union of token distributions that respect the viable reactions (see Appendix \ref{app:OEN_datalog}).

\section{Boolean matrix logic programming}
\label{sec:framework_bmlp}

We propose the Boolean Matrix Logic Programming (BMLP) approach, which uses Boolean matrices to evaluate datalog programs in contrast to the traditional symbolic logic program evaluation. 
\begin{definition} [Boolean Matrix Logic Programming (BMLP) problem]
    Let $\mathcal{P}$ be a datalog program containing a set of clauses with predicate symbol $r$. The goal of Boolean Matrix Logic Programming (BMLP) is to find a Boolean matrix $\textbf{R}$ encoded by a datalog program such that $(\textbf{R})_{i,j}$ = 1 if $\mathcal{P} \models r(c_i, c_j)$ for constant symbols $c_i, c_j$ and $(\textbf{R})_{i,j}$ = 0 otherwise.
\end{definition}

A linear and immediately recursive datalog program $\mathcal{P}$ can be written as matrices for bottom-up evaluation \cite{ceri_datalog_1989} and evaluated by linear equations \cite{sato_linear_2017}.
We created a BMLP algorithm\footnote{At the time of submission, the algorithm had been implemented in SWI-Prolog (version 9.2).} called iterative extension (BMLP-IE) (Fig. \ref{fig:bmlp_ie}) to evaluate the datalog encoding of the metabolic network. BMLP-IE uses a combination of binary operations: Boolean matrix addition (ADD), multiplication (MUL) and equality (EQ) \cite{copilowish_matrix_1948}. This algorithm uses $O(n^2)$ binary operations given $n$ total metabolites in the metabolic network (Appendix \ref{app:bmlp_ie}).

\begin{figure}[t]
    \centering
    \begin{tabular}{cc}
        \begin{subfigure} {0.55\textwidth}
            \includegraphics[width=\linewidth]{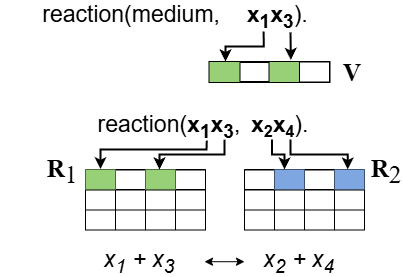}
            \caption{}
        \end{subfigure} & 
        \begin{subfigure} {0.35\textwidth}
            \includegraphics[width=\linewidth]{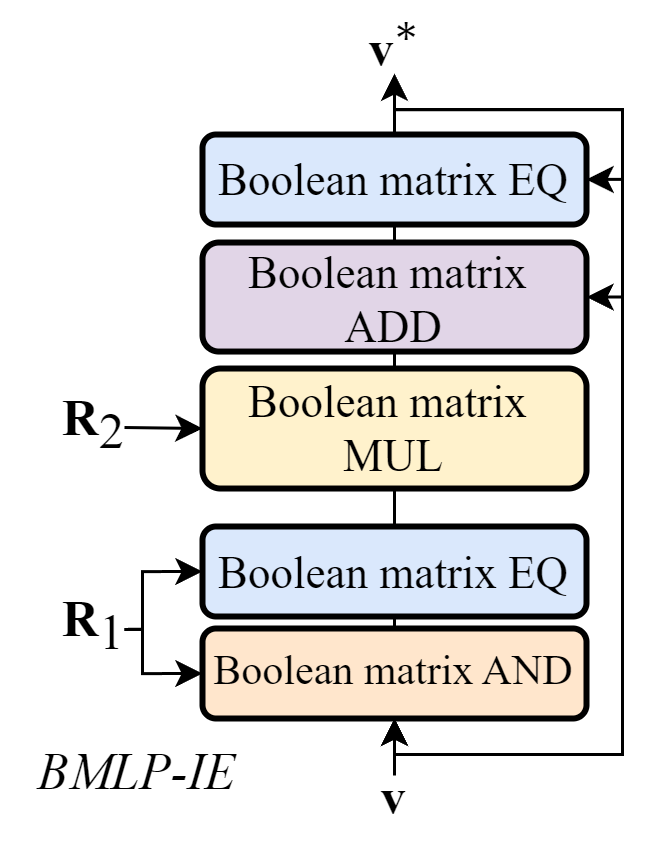} 
            \caption{}
        \end{subfigure} 
    \end{tabular}
    \caption{Iterative extension. (a) The vector $\textbf{v}$ encodes source chemical metabolites. All reactions are represented in the boolean matrices $\textbf{R}_1$ and $\textbf{R}_2$. (b) BMLP-IE computes $\textbf{v}^*$, the complete set of producible metabolites. }
    \label{fig:bmlp_ie}
\end{figure}

\section{Active learning}
\label{sec:framework_active_learning}
\subsection{Hypothesis compression}

We use the compression score of a hypothesis $h$ to optimise the posterior probability. A learner can derive the relative frequency $p(h)$ for $h$ to be chosen based on the encoding of $h$. The compactness of $h$'s encoding reflects its prior probability $p(h) = 2^{-size(h)}$ \cite{shannon_mathematical_1963}. We can use the compression function to compute the prior $p(h)$. Given a set of training examples $E$, the compression of a hypothesis \cite{progol,ase_progol} can be interpreted as 
\begin{flalign*}
    compression(h,E) = size(E) - MDL_{h,E} 
\end{flalign*}
In addition, the compression relates to the posterior probability of $h$ \cite{progol}:
\begin{flalign*}
    \frac{p(h|E)}{p(E|E)} = 2^{compression(h, E)}
\end{flalign*}
We can rearrange this to obtain the normalised posterior probability \cite{ase_progol}:
\begin{flalign*}
    p'(h|E) = \frac{2^{compression(h, E)}}{\sum_{h_i\in H} 2^{compression(h_i, E)}}
\end{flalign*}
The posterior probability is maximal when compression is maximal. By searching for a hypothesis with maximum posterior probability, a learner can maximise its expected predictive accuracy \cite{progol}. To learn an accurate hypothesis, the goal is to find a hypothesis with the highest compression score. 

The compression function in $BMLP_{active}$ follows the Minimal Description Length (MDL) principle  \cite{conklin_complexity-based_1994}. Since we have a fixed set of hypotheses without an explicit prior distribution, MDL allows us to make minimum assumptions about their prior. The most probable hypothesis $h$ for the training examples $E$ should be compact and have the fewest disagreements between $E$ and predictions made by $h$. This minimises $L(h) + L(E|h)$, where $L(h)$ is the descriptive complexity of $h$ and $L(E|h)$ is the descriptive complexity of encoding $E$ using $h$. We consider a compression function from \cite{ase_progol} based on the MDL principle:
\begin{flalign*}
    compression(h,E) = |E^+| - \frac{|E^+|}{pc_h} (size(h) + fp_h)
\end{flalign*}
where $E^+$ is the set of positive examples seen by the active learner, $pc_h$ is the number of positive predictions by $h$, and $fp_h$ is false positives covered by $h$. In other words, this compression function favours a general hypothesis with high coverage and penalises over-general hypotheses that incorrectly predict negative examples. The generality of $h$ can be estimated by computing its coverage $pc_h$ for a set of unlabelled instances \cite{progol}. We look at all available experiment instances to obtain a good estimation of the generality of hypotheses. Our BMLP approach can facilitate efficient inference for this estimation.  

\subsection{Expected cost of experiments}

Experiment planning is often restricted by time and other resources. A rule of thumb for experiment design is to achieve optimal empirical outcomes within a finite budget. The problem of designing experiments to support or refute scientific hypotheses involves a series of binary decisions. This process is analogous to playing the game of ``Guess Who'' where players ask a sequence of ``yes'' or ``no'' questions whose answers binarily partition the candidate hypotheses. The outcome of an experiment $t$ splits hypotheses $H$ into consistent hypotheses $H_{t}$ and inconsistent hypotheses $\overline{H_{t}}$. The path of selection and outcomes of experiments can be represented as a binary decision tree. To construct an optimal tree, an active learner should seek to label instances consistent with up to half of the hypotheses \cite{mitchell_generalization_1982}. 

In addition to this principle, we consider a user-defined experiment cost function $C_t$ of experiment $t$. The experiment selection is represented by a binary tree with paths annotated by the costs of experiments. The optimal selection is therefore a binary tree with the minimum expected overall cost. The minimum expected cost of performing a set of candidate experiments $T$ can be recurrently defined as:
\begin{flalign*}
    EC(\emptyset,T) = 0 & \\
    EC({h},T) = 0 & \\
    EC(H,T) = min_{t \in T} \, [C_t + p(t) EC (H_{t}, T - \{t\}) + (1 - p(t)) EC (\overline{H_{t}}, T - \{t\})]
\end{flalign*}
To estimate this recurrent cost function, $BMLP_{active}$ uses the following heuristic function from \cite{ase_progol} to approximate optimal cost selections:
\begin{flalign*}
    EC(H,T) \approx min_{t \in T} \, [C_t &+ p(t) (mean_{t' \in T - \{t\}} C_{t'}) J_{H_{t}} \\ \nonumber
    & + (1 - p(t)) (mean_{t' \in T - \{t\}} C_{t'}) J_{\overline{H_{t}}}]
\end{flalign*}
where $H_{t}$ and $\overline{H_{t}}$ are subsets of hypotheses $H$ consistent and inconsistent with $t$'s label. $p(t)$ is the probability that the outcome of the experiment $t$ is positive and $J_{H} = - \sum_{h \in H} p'(h|E) \, log_2 (p'(h|E))$. The probability $p'(h|E)$ is calculated from the compression function. To estimate $p(t)$, we compute the sum of the probabilities of the hypotheses consistent with a positive outcome of $t$. Users can flexibly define the experiment cost function $C_t$. 

\subsection{Active learning sample complexity}

For some hypothesis space $H$ and background knowledge $BK$, let $V_s$ denote the version space of hypotheses consistent with $s$ training examples. For an active version space learner that selects one instance per iteration, $|V_s|$ denotes the size of the version space at the iteration $s$ of active learning. The shrinkage of the hypothesis space can be represented by the reduction ratio $\frac{|V_{s+1}|}{|V_s|}$ after querying the $s+1$ label \cite{hocquette_how_2018}. The minimal reduction ratio $p(x_{s+1}, V_s)$ is the minority ratio of the version space $V_s$ partitioned by an instance $x_{s+1}$.

\begin{definition} (\textbf{Minimal reduction ratio} \cite{hocquette_how_2018}) The minimal reduction ratio over the version space $V_s$ by sampled instance $x_{s+1}$ is 
    \begin{flalign*}
        p(x_{s+1}, V_s) = \frac{min(\, |\{h \in V_s | \,h \cup BK \models x_{s+1}\}|, |\{h \in V_s | \,h \cup BK \not\models x_{s+1}\}|)}{|V_s|} 
    \end{flalign*}
\end{definition}

The minimal reduction ratio of an actively selected instance can be computed before it is labelled. While in reality training examples might not be as discriminative, the optional selection strategy is to select instances with minimal reduction ratios as close as possible to $\frac{1}{2}$ with the ability to eliminate up to 50\% of the hypothesis space. We describe the sample complexity advantage of active learning over random example selection by the following bound on an active version space learner's sample complexity. A passive learner is a learner using random example selection and it does not have control over the training examples it uses. Theorem \ref{theorem:active_learning} (proof in Appendix \ref{app:sample_complexity}) says that the number of instances needed to select by active learning should be some factor smaller than the number of randomly sampled examples given some predictive accuracy level. 

\begin{theorem}[Active learning sample complexity bound]
    For some $\phi \in [0, \frac{1}{2}]$ and small hypothesis error $\epsilon > 0$, if an active version space learner can select instances to label from an instance space $\mathcal{X}$ with minimal reduction ratios greater than or equal to $\phi$, the sample complexity $s_{active}$ of the active learner is 
    \begin{flalign}
        s_{active} \leq \frac{\epsilon}{\epsilon + \phi} s_{passive} + c
        \label{eq:sample_complexity_bound}
    \end{flalign}
    where c is a constant and $s_{passive}$ is the sample complexity of learning from randomly labelled instances.
    \label{theorem:active_learning}
\end{theorem}

Given a desirable predictive error $\epsilon$, the factor in Theorem \ref{theorem:active_learning} can be estimated from the minimal reduction ratio $\phi$ over the instance space. When selected instances have a larger minimal reduction ratio, the sample complexity gain would be more significant. On the other hand, if instances have low discriminative power and the minimal reduction ratio tends to zero, the sample complexity gain would be minimal. In Section \ref{sec:experiments}, we empirically show evidence supporting this theorem for $BMLP_{active}$ and randomly sampled experiments when learning gene function annotations in a GEM.

\section{Implementation}
\label{sec:implementation}

\begin{algorithm}[h]
    \caption{$BMLP_{active}$}
    \label{alg:active_learning}
    \textbf{Input}: Hypotheses $H$, background knowledge $BK$, experiment instances $T$, experimental cost budget $C$, total number of experiment $N$ and oracle $O$.\\
    \textbf{Output}: A hypothesis $h_{max} \in H$.
    \begin{algorithmic}[1] 
        \STATE Let $T_{selected} = \emptyset$, $V_0 = H$
        \WHILE{$cost(T_{selected}) \leq C$ and $|V_j| > 1$ and $|T_{selected}| < N$}
            \IF{$T_{selected} == \emptyset$}
                \STATE Select $t_j$ with the largest minimum reduction ratio $p(t_j,V_j)$
                \STATE \,\, from remaining experiments in $T$ to add to $T_{selected}$
            \ELSE
                \STATE Add experiment $t_j \in T$ that has the minimum cost $EC(V_j, T')$ to $T_{selected}$.
            \ENDIF
                \STATE break \textbf{if} $cost(T_{selected}) > C$
                \STATE Consult the label $l$ of $t_j$ from the oracle $O$
                \STATE $\textbf{R}_{i,j}$ = $BMLP\_phenotype\_prediction$($h_i \cup BK$, $t_j$) for all $h_i \in V_j$
                \STATE $V_{j+1} = \{h_i \, | \, h_i \in V_j \text{ and } \textbf{R}_{i,j} = l\}$
                \STATE $h_{max} = \mathrm{argmax}_{h_i \in V_{j+1}} \, compression(h_i,T_{selected})$
        \ENDWHILE
    \end{algorithmic}
\end{algorithm}

The input background knowledge is created from a GEM model, which is typically accessible from GEM repositories \cite{king_bigg_2016,li_gotenzymes_2023} and libraries \cite{schellenberger_quantitative_2011,ebrahim_cobrapy_2013}. The user provides a set of candidate experiments, hypothesis candidates and an experiment cost function. Labels of experiment instances are requested from a data source, e.g. a laboratory or an online dataset. Labelled experimental data are considered ground truths. Initially, we randomly shuffle all candidate experiments and select an experiment with the most discriminative power if its cost is under the budget. Alternatively, a discriminative experiment with the lowest cost can be selected for a fixed initialisation. 

 Active learning in $BMLP_{active}$ is described by Algorithm \ref{alg:active_learning}. A classical method to address errors in metabolic pathways involves auxotrophic growth experiments, where specific genes are deleted to render the organism incapable of synthesising essential compounds \cite{beadle_genetic_1941}. We obtain binary phenotypic classifications of cell growth comparable to an unedited wild-type strain. We identified all essential metabolites in the wild-type strain from the BMLP-IE output. A positive label (1) is a phenotypic effect, and a negative label (0) is no phenotypic effect. We predict labels from the combinations of candidate experiments $T$ and hypotheses $H$ and store them in a Boolean matrix $\textbf{R}$ of size $|H| \times |T|$. Given the j-th actively selected instance $t_j \in T$ and $h_i \in V_{j} \subseteq H$, $(\textbf{R})_{i,j} = 1$ if $h_i \cup BK \models t_j$ and otherwise $(\textbf{R})_{i,j} = 0$. 

In each active learning cycle in $BMLP_{active}$, hypotheses inconsistent with ground truth experimental outcomes are pruned. $BMLP_{active}$ selects experiments to minimise the expected value of a user-defined cost function. This avoids repetitive predictions in hypothesis pruning and compression calculation. When computing hypothesis compression, we estimate the generality of a hypothesis by calculating its coverage for all unlabelled instances based on $\textbf{R}$. Currently, we only consider labels from synthetic or online phenotype data. However, $BMLP_{active}$ can be coupled with a high-throughput experimental workflow to automate experiments. 

\section{Experiments}
\label{sec:experiments}

\subsection{Experiment 1: BMLP-IE runtime and predictions}

\begin{table}[t]
\centering
	\begin{tabular}{c|c|c}
		& Single thread (seconds) & 20 CPUs (seconds) \\ & & \\
		\hline
		& & \\
		SWI-Prolog & $37.842 \pm 9.668$ & $5.089 \pm 0.119$\\ 
		& & \\ \hline
		& & \\
		\textbf{SWI-Prolog + BMLP-IE} & $\mathbf{0.220 \pm 0.099}$ & $\mathbf{0.061 \pm 0.008}$ \\ 
	\end{tabular}
\bigskip
\caption{Mean runtime of 100 predictions in CPU time. BMLP-IE leads to a 170 times improvement in prediction time efficiency. Multi-threading BMLP-IE enhances runtime efficiency 600 times compared to base SWI-Prolog. }
\label{table:runtime_improvement}
\end{table}

Here we demonstrate\footnote{All experiments are performed on a workstation with Intel(R) Core(TM) i9-7900X CPU \@ 3.30GHz (20 CPUs). } that a SWI-Prolog application can be made significantly faster by involving Boolean matrix computation. We randomly sampled batches of 100 experiments and computed the average wall time from 10 repeats (Appendix \ref{app:experiment1}). Table \ref{table:runtime_improvement} shows a 170 times improvement in prediction time by BMLP-IE compared to just using just SWI-Prolog. BMLP-IE uses only $\frac{1}{600}$ of the prediction time compared to predictions with just SWI-Prolog. This result demonstrates that BMLP-IE significantly improves the speed in predicting phenotypes from a genome-scale model. 

\begin{figure}[t]
\centering
\includegraphics[width=\textwidth]{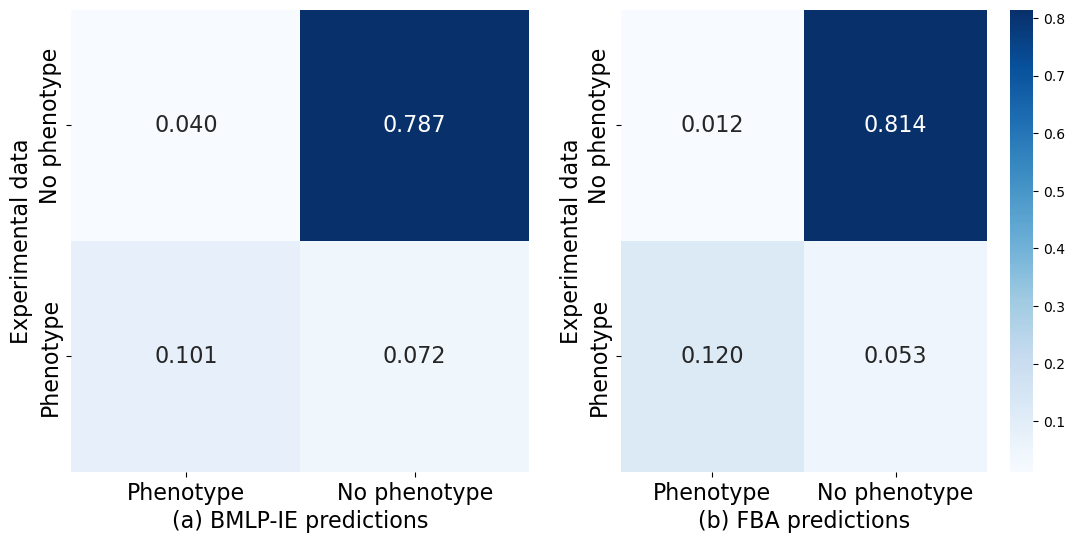}
\caption{Normalised confusion matrices. The predictive accuracy of BMLP-IE is $88.7\%$ and standard FBA is $93.4\%$. Each number has been normalised by the size of the experiment dataset in \cite{iML1515}. }
\label{fig:confusion_matrix}
\end{figure}

FBA \cite{palsson_systems_2015} is widely used in the literature to quantitatively predict the cell growth rate based on fluxes of biomass metabolites at the steady state. We evaluated BMLP-IE's phenotypic predictions against FBA based on the GEM model iML1515 \cite{iML1515}. We created normalised confusion matrices (Fig. \ref{fig:confusion_matrix}) for BMLP-IE and the FBA predictions with respect to experimental data \cite{iML1515}. The confusion matrices show a comparable predictive accuracy by BMLP-IE to that of standard FBA. Compared with FBA's growth rate predictions, BMLP-IE mainly misclassified reduced cell growth. BMLP-IE is easily adaptable when  GEM is extended, but FBA requires extensive manual parameter updates in order to tune metabolic fluxes.

\subsection{Experiment 2: active learning sample complexity and experiment cost}
\label{experiment2}

\begin{figure}[t]
\centering
\includegraphics[width=0.8\textwidth]{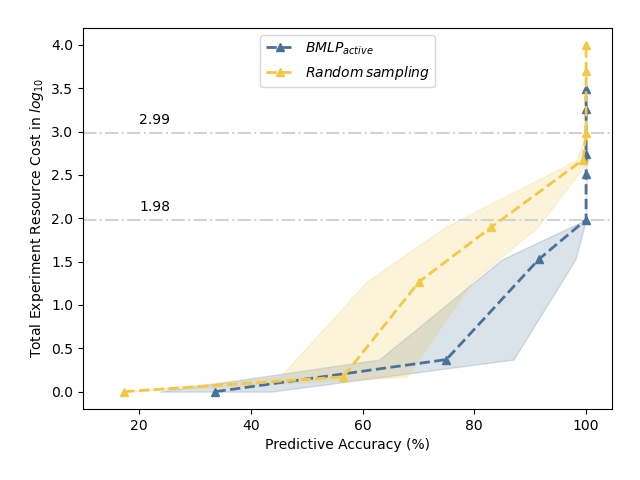}
\caption{Total experimentation cost needed for the recovered models to reach different levels of predictive accuracy. The cost definitions and details are in Appendix \ref{app:experiments} and \ref{app:appendix_cost}). The predictive accuracy of an empty hypothesis is 17.4\%.}
\label{fig:cost_reduction}
\end{figure}

We examined how selecting experiments intelligently improves data and cost efficiency in contrast to selecting experiments randomly\footnote{$BMLP_{active}$ actively selected experiments from the instance space to reduce the hypothesis space. The random selection method randomly sampled instances from the instance space. }. The number of experiments and data required are major determinants of experimental cost. We obtained synthetic data from iML1515 (Appendix \ref{app:experiment2}), and then masked function annotations in iML1515 to mock actual discoveries. We applied $BMLP_{active}$ and random experiment selection with compression to decide the final hypothesis for recovering iML1515. We tested the predictive accuracy of the recovered function hypotheses against synthetic data. Each selection method was repeated 10 times. 

Fig. \ref{fig:cost_reduction} shows that $BMLP_{active}$ spends 10\% of the cost used by random experiment sampling when converging to the correct hypothesis. The y-axis is the total experimental reagent cost in $log_{10}$ calculated from selected experiments and normalised by the cost of the cheapest optional nutrient. A predictive accuracy of 100\% indicates successful recovery of deleted gene function annotations. The correct annotations were recovered by $BMLP_{active}$ and random experiment selection with $10^{1.99}$ and $10^{2.99}$ experimental reagent costs. Given the current hypothesis space, experimental design space and our assumption that one gram per reagent was used, learning a single annotation required $\pounds 3.8$ and $\pounds 38$ experimental costs from $BMLP_{active}$ and random sampling. Consideration of differential resource factors in the user-defined experimental cost function could significantly increase the total experimental cost.

In addition, Fig. \ref{fig:sample_complexity_reduction} shows that $BMLP_{active}$ reduces the number of labelled data to learn accurate gene function annotations, compared to random experiment selection. While random experiment selection requires 25 experiments to recover the deleted annotations, $BMLP_{active}$ only needed 3 experiments. Fig. \ref{fig:cost_reduction} and Fig. \ref{fig:sample_complexity_reduction} show $BMLP_{active}$ can simultaneously optimise a user-defined experimental cost function and reduce the total number of experiments to learn gene function annotations accurately. $BMLP_{active}$ is highly flexible and could be tailored for specific experimental objectives in a discovery process.

\begin{figure}[t]
\centering
\includegraphics[width=0.9\textwidth]{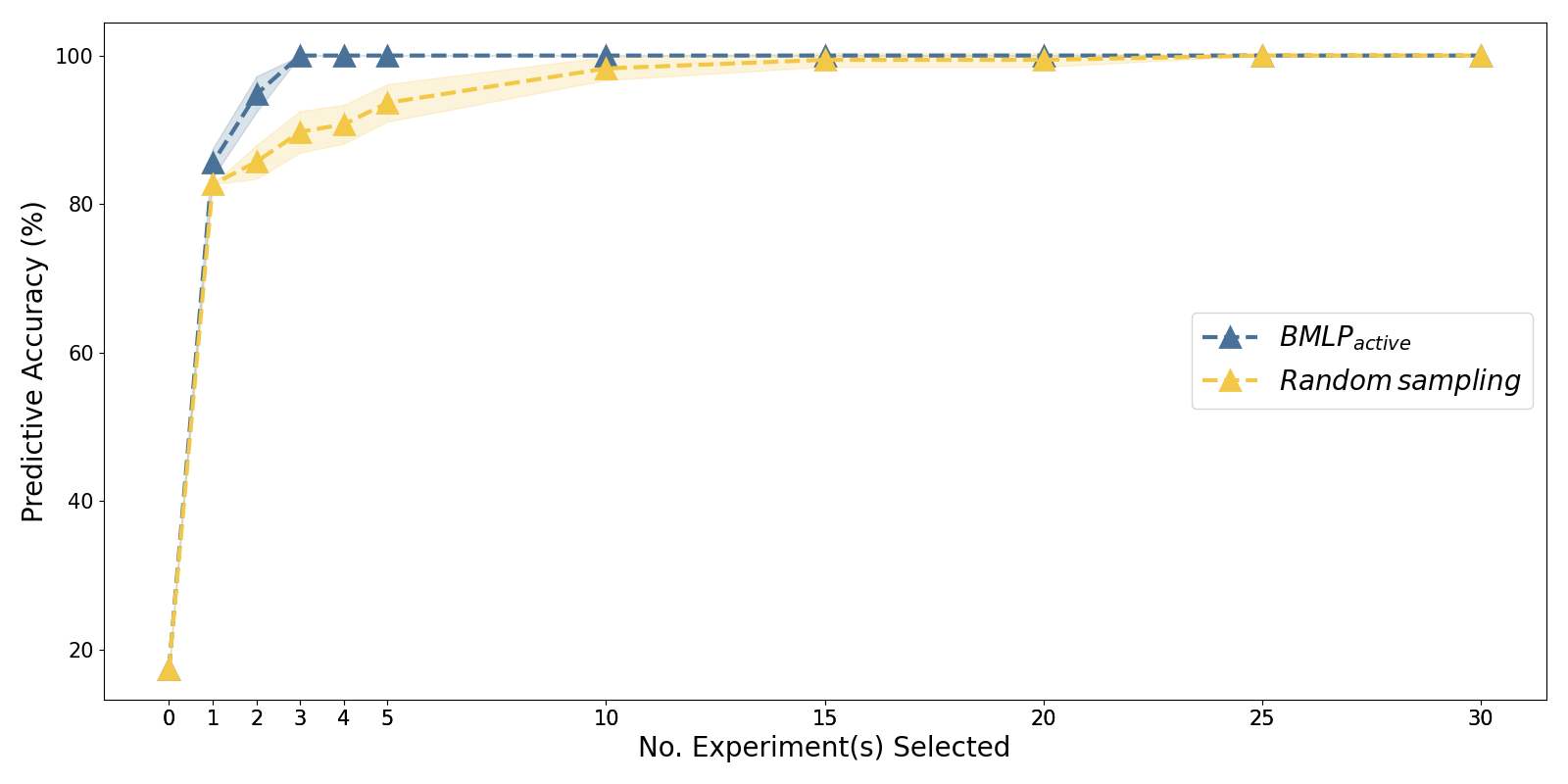}
\caption{Reduction of the number of experiments needed to recover gene function annotations (sample complexity reduction). $BMLP_{active}$ learns fully accurate gene function annotations with 3 experiments, while random sampling requires 25 experiments to do so. The predictive accuracy of an empty hypothesis is 17.4 \%. Predictive accuracy is higher for both methods initially since experiment selection is no longer constrained by reagent costs.}
\label{fig:sample_complexity_reduction}
\end{figure} 

\begin{figure}[t]
\centering
\includegraphics[width=\textwidth]{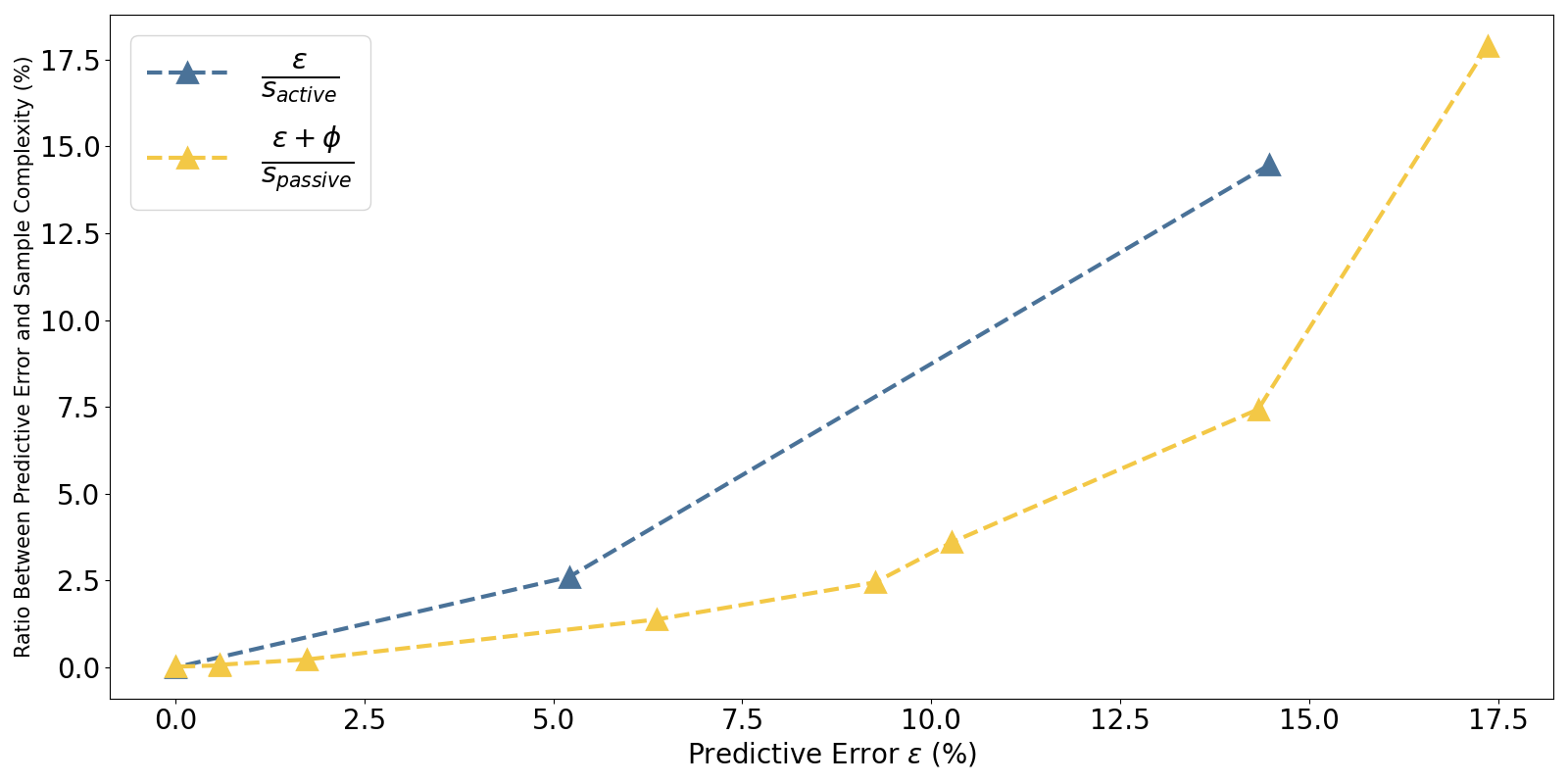}
\caption{Ratio of sample complexity. $BMLP_{active}$ selects instances with a minimal reduction ratio of at least 0.5\%. For error $\epsilon > 0$, the ratio of $BMLP_{active}$ is clearly above the ratio of random sampling. This shows that Theorem \ref{theorem:active_learning} correctly describes the sample complexity relationship between active learning and random sampling.}
\label{fig:error_ratio}
\end{figure} 

Fig. \ref{fig:error_ratio} shows the relationships between predictive error and the number of experiments in Theorem \ref{theorem:active_learning}. $BMLP_{active}$ selected experiments with minimal reduction ratio $\phi \geq 0.5\%$. The ratio $\frac{\epsilon}{s_{active}}$ is greater than $\frac{\epsilon + \phi}{s_{passive}}$ for error $\epsilon > 0$. There are fewer points on the active learning ratio curve since $BMLP_{active}$ converges to accuracy 100\% using fewer training examples. The two curves intersect at $\epsilon = 0$ since both converged to 100\% accuracy. Theorem \ref{theorem:active_learning} shows that active learning ($BMLP_{active}$) can reduce the requirement for training examples compared with passive learning (random experiment selection) to achieve a target learning performance. In the context of scientific discovery, the experimental cost to arrive at a finding can be reduced since each actively selected experiment is more informative.  

\subsection{Experiment 3: learning digenic functions}

\begin{figure}[t]
\centering
\includegraphics[width=0.9\textwidth]{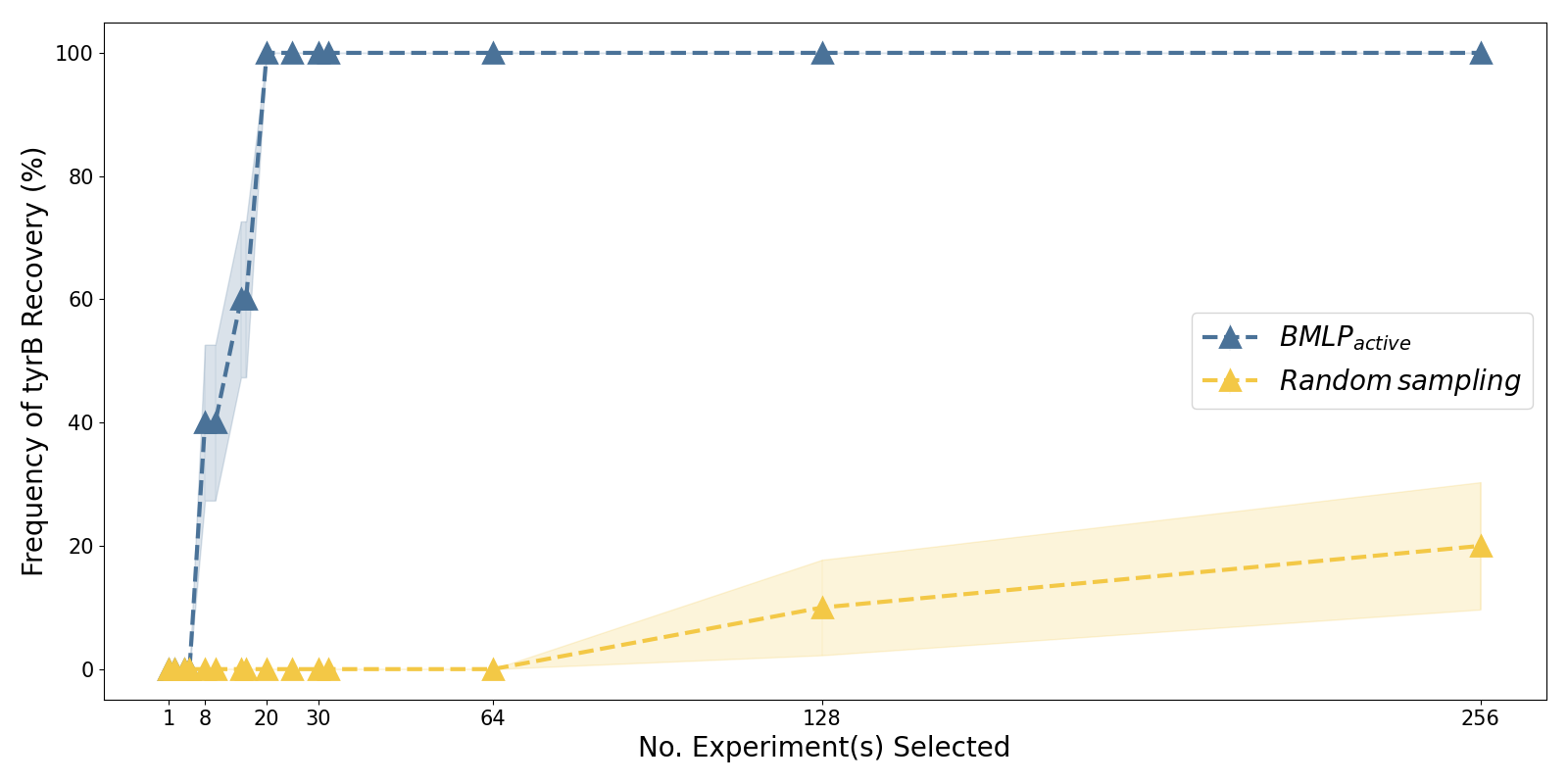}
\caption{tyrB isoenzyme function recovery frequency. We obtain the frequency of successful recovery from various randomisation seeds and compute the standard errors. We observe that within 64 experiments, random experiment selection does not recover the tyrB isoenzyme function. In contrast, $BMLP_{active}$ can recover this function with high frequency with 20 experiments. }
\label{fig:isoenzyme_recovery}
\end{figure} 

We show the recovery of a known gene-isoenzyme association, which presents a bigger empirical challenge due to the need to explore a larger experimental design space. We focused on the aromatic amino acid biosynthesis pathway. It is a critical pathway for microorganism survival that contains important gene targets for engineering aromatic amino acid production strains. Although this pathway has been intensively studied for decades, gaps remain in the understanding of gene-gene relations \cite{price_filling_2018}. Fig. \ref{fig:isoenzyme_recovery} shows the result of recovering tyrB's isoenzyme function. We observed that $BMLP_{active}$ successfully recovers the correct tyrB isoenzyme function in 20 experiments. This demonstrates $BMLP_{active}$'s ability to discover gene-isoenzyme associations in GEM. However, to learn all isoenzyme functions, a high-throughput experimentation procedure is required and we further discuss this in Section \ref{sec:conclusion}. 

$BMLP_{active}$ significantly reduced the number of experiments needed compared to random experiment selection. $BMLP_{active}$ provides higher information gain from each experiment and can guarantee recovery with as few as 20 experiments. While random sampling could prune most candidate hypotheses, it failed to eliminate competitive hypotheses further. In contrast to $BMLP_{active}$, random sampling could not recover this isoenzyme function with 64 experiments and could only do this occasionally with more than 10 times the number of experiments required by $BMLP_{active}$. Although we used a reduced set of genes in this demonstration, the combinatorial space was sufficiently large that random experiment selection became non-viable as an experimentation strategy in a discovery process. The remarkable increase in efficiency with $BMLP_{active}$ with this task demonstrates it has potential even as the experimental design space grows. \\

\section{Related work}
\label{sec:related_work}

\subsection{Active learning}

A typical approach to active learning is via membership queries, where the machine learner queries a data source to label instances according to some strategy. An active learner can be implemented to explore the version space which contains all hypotheses that are expressible by a hypothesis language and are consistent with a set of training examples \cite{mitchell_generalization_1982}. Instances are chosen according to a measure of informativeness such as entropy \cite{shannon_mathematical_1963}, diameter \cite{tosh_diameter_based_2017}, the size \cite{mitchell_generalization_1982,dasgupta_coarse_2005} and the shrinkage \cite{hocquette_how_2018} of the version space for eliminating competing hypotheses. For a binary classification task, when learners are allowed to look at $n$ unlabelled instances to compare, an active learner has $n$ times larger probability of selecting an instance with maximal entropy than a passive learner \cite{hocquette_how_2018}. The general principle for a binary active learner is to label instances closest to being consistent with half of the hypotheses in the version space \cite{mitchell_generalization_1982}. These instances reject up to half of the hypothesis space regardless of their classifications. In the optimal case when all instances are maximally discriminative, the target hypothesis can be found via a binary search with depth logarithmic of the hypothesis space size \cite{angluin_queries_2001}. $BMLP_{active}$ makes membership queries to simultaneously reduce the number of labels and the expected cost of experimentation. 

\subsection{Computational scientific discovery}

Computational discovery systems \cite{langley_scientific_1987,todorovski_integrating_2006,brunton_discovering_2016,guimera_bayesian_2020,petersen_deep_2020} are mostly responsible for formulating symbolic hypotheses from experimental results, which are not directly applicable for experimental planning. The automation of experimental design was investigated by the pioneering Robot Scientist \cite{King04:RobotScientist} to bring together logical reasoning, active learning and laboratory automation. The Robot Scientist automatically proposes explanatory hypotheses, actively devises validation experiments, performs experiments using laboratory robotics, and interprets empirical observations. It was combined with a laboratory robot to learn gene function annotations for aromatic amino acid pathways, and the active learning strategy significantly outperformed random experiment selection \cite{King04:RobotScientist}. This was extended to identify genes encoding orphan enzymes from a logical model of yeast consisting of 1200 genes \cite{king_automation_2009}. In comparison, we apply $BMLP_{active}$ to a genome-scale metabolic network, which is significantly larger than the models previously examined. In addition, we examined isoenzymes, which were previously unexplored by logic programming due to larger hypothesis and experimental design spaces.

Recently, large language models (LLMs) have been explored to automate experimental design via user-provided prompt texts \cite{boiko_autonomous_2023,bran_augmenting_2023}. LLMs are generative models of natural languages based on neural networks that are trained from textual data. In \cite{boiko_autonomous_2023,bran_augmenting_2023}, LLMs served as experiment planners and coding assistants. These platforms use LLMs to access publicly available knowledge bases, write and interpret programs to control laboratory hardware and analyse collected experimental data. The advantage is accessing rich information online and interacting with users while designing experiments. However, the faithfulness of their outputs raises concerns. It has been stressed that AI autonomous experimentation should be treated with caution, when there are no means to identify the source of error and explain observations \cite{ren_autonomous_2023}. $BMLP_{active}$'s logical representation and computation enhance faithfulness and output verifiability, which do not suffer from the ambiguity of natural languages. 

\subsection{Learning genome-scale metabolic network models}
GEMs contain known metabolic reactions in organisms, serving as functional knowledge bases for their metabolic processes. There has been increasing interest in integrating machine learning methods with GEMs to improve the predictions of gene-phenotype correlations within biological systems \cite{angione_human_2019,sen_deep_2021}. This often requires constraint-based modelling techniques for simulating GEMs \cite{rana_recent_2020}. Flux balance analysis (FBA) \cite{palsson_systems_2015} is a widely used mathematical tool to solve flux distributions in genome-scale models under the steady-state condition. FBA uses linear programming to model fluxes in a chemical reaction network represented by a matrix that contains the stoichiometric coefficients of the metabolites in reactions. Ensemble learning frameworks such as \cite{wu_rapid_2016,oyetunde_machine_2019,heckmann_machine_2018} obtained support vector machines, decision trees and artificial neural networks to represent hidden constraints between genetic factors and metabolic fluxes. These systems cannot identify inconsistencies between the simulation results and the experimental data to self-initialise the constraints to aid learning \cite{rana_recent_2020}. On the other hand, mechanistic information that reflects whole-cell dynamics has been incorporated to address these concerns, for instance in artificial neural networks to tune their parameters \cite{yang_white_box_2019}. The hybrid approach in \cite{faure_neural-mechanistic_2023} embedded FBA within artificial neural networks based on custom loss functions surrogating the FBA constraints. Recent research also explored autoencoders to learn the underlying relationship between gene expression data and metabolic fluxes through fluxes estimated by FBA \cite{sahu_advances_2021}. Machine learning of genome-scale metabolic networks has been significantly limited by the availability of training data \cite{sen_deep_2021}. $BMLP_{active}$ does not use FBA, so it is not dependent on metabolic flux constraints. $BMLP_{active}$ also performs active learning to reduce the cost and number of experiments. This is a notable advantage considering finite experimental resources and insufficient experimental data.

\subsection{Matrix algebraic approaches to logic programming} 

Obtaining the least Herbrand model of recursive datalog programs can be reduced to computing the transitive closure of Boolean matrices \cite{peirce_collected_1932,copilowish_matrix_1948}. Fischer and Meyer \cite{fischer_boolean_1971} studied a divide-and-conquer Boolean matrix computation technique by viewing relational databases as graphs. A similar approach was explored by Ioannidis \cite{ioannidis_computation_1986} for computing the fixpoint of recursive Horn clauses. An ILP system called DeepLog \cite{muggleton_hypothesizing_2023} employed Boolean matrices for learning recursive datalog programs. It repeatedly computes square Boolean matrices to quickly choose optimal background knowledge to derive target hypotheses. A propositional and Boolean matrices representation of biological knowledge bases was explored in \cite{sato_boolean_2021,sato_differentiable_2023} by considering gene regulatory networks as Boolean networks. In our BMLP framework, we use Boolean matrices to evaluate first-order recursive datalog programs. 

Recent work primarily studied the mapping of logic programs to linear equations over tensor spaces. SAT problems were investigated \cite{lin_satisability_2013} under a linear algebra framework, where truth values of domain entities, logical relations and operators for first-order logic are evaluated using tensors-based calculus \cite{grefenstette_towards_2013}. Computing algebraic solutions can approximate recursive relations in first-order Datalog \cite{sato_linear_2017}. Based on this approach, abduction can be performed by encoding a subset of recursive datalog programs in tensor space \cite{sato_abducing_2018}. Tensor representations allow differentiable infrastructures such as neural networks to perform probabilistic logical inferences \cite{cohen_tensorlog_2020}. This can be done using dynamic programming to propagate beliefs for stochastic first-order logic programs. The probabilistic parameters can be optimised using deep learning frameworks \cite{yang_differentiable_2017}. Limitations have been identified for solving certain recursive programs due to the difficulty of computing some arithmetic solutions, whereas these can be solved by iterative bottom-up evaluations \cite{sato_linear_2017}. In addition, algebraic approximations are usually calculated in linear algebraic methods where the correctness of the solutions cannot be guaranteed. However, our BMLP approach does not approximate Boolean matrices. 

\section{Conclusion and future work}
\label{sec:conclusion}

To our knowledge, $BMLP_{active}$ is the first logic programming system able to operate on a state-of-the-art GEM, utilising the improved computational efficiency of Boolean matrices. $BMLP_{active}$ learned accurate gene function annotations while reducing the experimental costs by 90\%. In addition, we demonstrated $BMLP_{active}$ actively learned gene-isoenzyme functions with 90\% fewer training data points than random experiment selection. This suggests that $BMLP_{active}$ could offset the growth in experimental design space from genome-scale function learning. Future work could explore the use of Boolean matrix encoding of logic programs for even larger networks with GPU acceleration.

One limitation of $BMLP_{active}$ is that it cannot quantitatively predict the behaviours of biological systems. Future work could focus on using Probabilistic Logic Programming \cite{de_raedt_probabilistic_2015} and tensor-based algebra of logic programs \cite{sakama_logic_2021} given numeric training data. $BMLP_{active}$ also does not invent new predicates. Reusing invented predicates can reduce the size of a hypothesis learned by an ILP system and improve its learning performance \cite{ILP30}. Recent methods \cite{muggleton_hypothesizing_2023,evans_learning_2018,dai_abductive_2021} have enabled predicate invention in the context of matrix-driven ILP. While we touched on the theoretical properties of active learning, comparisons with other active learning approaches and theoretical claims can further contextualise our performance.

Importantly, we show that our approach can address the interactions between multiple genes, which are not addressable by current mutagenesis approaches. We aim to integrate $BMLP_{active}$ into a high-throughput experimental workflow \cite{wang_pooled_2018} for whole-genome function learning. Experimental data can first be compared with phenotypic predictions, narrowing the scope of the hypothesis space. With advanced biological gene editing tools, we can target specific gene edits to experimentally confirm candidate hypotheses during active learning.

Designing experiments to explain some hypothesised phenomenon requires deliberate thinking by scientists. A helpful scientific assistant could teach scientists how to create experiment plans with optimal outcomes. This aspect of human-AI collaboration has been studied in an emerging research area on Ultra-Strong Machine Learning (USML) \cite{US2018,Ai2021,sequential_teaching,krenn_scientific_2022}. The present work can be extended to examine the effect of AI decision-making on human experiment selection and human comprehension of complex biological concepts.

\section*{Ethical Statement}
The datasets used do not contain any personal information or other sensitive content. There are no ethical issues.

\section*{Data Availability}
Source code and data are accessible on \\ https://github.com/lun-ai/BMLP\_active/. 

\section*{Acknowledgments}
The first, third and fourth authors acknowledge support from the UKRI 21EBTA: EB-AI Consortium for Bioengineered Cells \& Systems (AI-4-EB) award \\(BB/W013770/1). The second author acknowledges support from the UK’s EPSRC Human-Like Computing Network (EP/R022291/1), for which he acts as Principal Investigator. ChatGPT and Grammarly have been used as language editing tools after all intellectual content has been drafted.

\vspace{.5em}
\large
\bibliographystyle{abbrv}
\bibliography{active_abduction,bmlp}

\setcounter{section}{0}
\begin{appendix}

\section{One-bounded elementary net (OEN)}
\label{app:OEN}
Petri nets have wide applications in Computer Science \cite{reisig_understanding_2013}. Elementary nets \cite{rozenberg_elementary_1998} are a class of Petri nets. We summarise key concepts and notions of Petri net and elementary nets from \cite{reisig_understanding_2013,rozenberg_elementary_1998}. A Petri net contains two disjoint sets of nodes: places $P$ and transitions $T$. The flow relation $F \subseteq (P \times T) \cup (T \times P)$ describes arcs between nodes. All places can be marked by tokens which symbolise the resources distributed in the network. The assignment of tokens to places is called a marking (of tokens). For a transition $t$, the pre-set $^{\bullet}t$ = $\{ p \in P \, | \, pFt\}$ contains places whose tokens would be consumed by $t$ and its post-set $t^{\bullet}$ = $\{ p \in P \, | \, tFp\}$ contains places that would receive tokens from $t$. Markings represent the states of a Petri net. The weight of an arc $(p, t)$ or $(t, p)$ is conventionally written as $\overline{pt}$ or $\overline{tp}$. Elementary nets are basic Petri nets where transitions consume one token from the pre-set and relocate a token to the post-set. In an elementary net, an arc weight $\overline{pt} = 1 \text{ if } (p,t) \in F$ and $\overline{tp} = 1 \text{ if } (t,p) \in F$. An elementary net is a tuple $(P, T, F)$. 

A one-bounded elementary net (OEN) is a class of elementary nets that has places marked with at most one token \cite{reisig_understanding_2013}. A marking $m$ in OEN is a characteristic function $m: P \to \{0,1\}$ defined as:
\begin{gather}
m(x) = \left\{
\begin{aligned}
& 1 \qquad\,\,\,\,\, \text{if a token is assigned to } x \\ 
& 0 \qquad\,\,\,\,\, \text{otherwise} \nonumber
\end{aligned}
\right.
\end{gather}
A place $p$ is marked if $m(p) = 1$. An initial marking corresponds to the initial state of an elementary net and its tokens. A transition $t$ is enabled by $m$ if and only if for all arcs $(p,t)\in F$, $m(p) = 1$. A step $m \xrightarrow[]{t} m'$ is a fired transition resulting in a new marking according to: 
\begin{gather}
m'(p) = \left\{
\begin{aligned}
& m(p) - 1 \qquad\,\,\,\,\, \text{if } (p, t) \in F \\ 
& m(p) + 1 \qquad\,\,\,\,\, \text{if } (t, p) \in F \\
& m(p) \qquad\qquad\,\,\,\,\,\, \text{otherwise}
\end{aligned}
\right.
\label{eq:step_update}
\end{gather}
A marking $m'$ is reachable from a marking $m$ if and only a finite step sequence exists. A place is reachable from $m$ if and only if $m'$ is a reachable marking from $m$ and $m'(p) = 1$. A marking can be written as a row vector ${[} 1, 0, 1, ... {]}$ where each element indicates if the respective place has a token. A transition can also be represented by a row vector ${[} z_1, z_2, ..., z_k {]}$ where $z_i$ reflects the assignment of tokens at the pre-set and post-set places. The consecutive firings of transition update markings and the global state of the system. This represents the dynamic changes in the distribution of tokenised resources. 

\section{Equivalence between OEN and datalog}
\label{app:OEN_datalog}
Each auxotrophic mutant experiment is treated as a datalog query to find all reachable places based on an initial marking. We evaluate queries of OEN reachability by transforming the OEN into a linear and immediately recursive datalog program. A datalog program consisting of clauses with arity two can be considered a directed graph where edges joining nodes represent the relation between terms. Places and transitions are represented as constant and predicate symbols. Each argument of our transformed datalog program is a concatenation of constant symbols, e.g. $h2\_o2$. For our specific purpose, since all place nodes in the OEN are related to metabolites, every new concatenated constant symbol describes a set of metabolic substances. 

\begin{definition} [\textbf{OEN transformation}]
    The transformation takes an OEN and an initial marking $m$. Let $t$ denote any transition. Input places of transitions are represented by $p^*$. A datalog program $\mathcal{P}$ is written with predicate symbols $t$, $r_1$ and $r_2$ and variables X, Y and Z:
    \begin{itemize}
        \item write $token(m, p^*)$ if the marking $m$ includes all $p^*$
        \item write $token(m, q)$ if the marking $m$ includes a place $q$
        \item write $t(p_1^*,p_2^*)$ if $t$ connects $p_1^*$ and $p_2^*$
        \item write $t(p^*, q)$ if $t$ connects $p^*$ and a place $q$
        \item write $r_1(X,Y) \gets token(X,Y)$ and $r_1(X,Y) \gets t(X,Y)$ for every $t$ clause written 
        \item write $r_2(X,Y) \gets r_1(X, Y)$ and $r_2(X,Y) \gets r_1(X, Z), r_2(Z, Y)$
         \nonumber
    \end{itemize}
    \label{def:transformation}
\end{definition}
The transformed program is linear and immediately recursive since it contains only one recursive clause, and the recursive predicate appears only once in the head of this clause. We illustrate how reachable places from an initial marking can be computed by evaluating the transformed program. 
\begin{example}
\textit{We use the following OEN to represent a metabolic pathway. Transitions $t_1$, $t_2$ and $t_3$ are described by $reaction\_1$, $reaction\_2$ and $reaction\_3$ clauses. We create two nodes $c1\_c2$ and $c3\_c4$. These transitions connect places $\{c_3$, $c_4$, $c_5\}$ which are represented by constants $\{c3$, $c4$, $c5\}$. For an initial marking $m_1 = \{c_1, c_2\}$, we can construct a datalog program $\mathcal{P}_2$:}\\
\begin{minipage}{0.4\textwidth}
    \centering
    \includegraphics[width=0.9\linewidth]{figs/elementary_net_2.png}
    \label{fig:enter-label}
\end{minipage} 
\begin{minipage}{0.6\textwidth}
\begin{gather}
\mathcal{P}_2:\left\{
\begin{aligned} 
    token(m1,c1\_c2). &\\
    token(m1,c1).&\\
    token(m1,c2).&\\
    reaction\_1(c1\_c2, c3\_c4). &\\
    reaction\_1(c1\_c2, c3). &\\
    reaction\_1(c1\_c2, c4).&\\
    reaction\_2(c3\_c4, c5). &\\
    reaction\_3(c5, c4). &\\
    reaction(X,Y) \gets token(X,Y). &\\
    reaction(X,Y) \gets reaction\_1(X,Y). &\\
    reaction(X,Y) \gets reaction\_2(X,Y). &\\
    reaction(X,Y) \gets reaction\_3(X,Y). &\\
    pathway(X,Y) \gets reaction(X,Y). &\\ 
    pathway(X,Y) \gets reaction(X,Z), &\\
                pathway(Z,Y). & \nonumber
\end{aligned}
\right\}
\end{gather} 
\end{minipage}
\bigskip

\textit{The right-hand-side figures are the OEN represented by $\mathcal{P}_2$. The node $\{c_3, c_4\}$ is reachable from $\{c_1, c_2\}$ via the transition $t_1$. The place $c_5$ is reachable in two steps via transitions $t_1$ and $t_2$. We can compute these by evaluating groundings of $pathway(m1, Y)$ in the Herbrand model $M(\mathcal{P}_2)$ which gives us $pathway(m1,c2)$, $pathway(m1, c3\_c4)$ and $pathway(m1, c5)$. }
\label{ex:transformation}
\end{example}

We prove the correctness of computing reachable places in an OEN by evaluating the transformed datalog program. We first show the effect of a step (a fired transition) on updating markings. A step in an OEN updates the set of marked places. 

\begin{lemma}
    Let $OEN = (P,T,F)$. For a transition $t \in T$, marking $m$ and $m'$, $\land_{p_i\in ^{\bullet}t} m(p_i) = 1\to \land_{p_j \in t^{\bullet}} m'(p_j) = 1$.
    \label{lemma:step}
\end{lemma}
\begin{proof}
    It follows that all tokens at the places in the pre-set of a transition are consumed, and places in the post-set of a transition are assigned tokens.
\end{proof}
Lemma \ref{lemma:step} shows that when a transition $t$ fires, in the new marking, tokens are re-assigned to other places. Given a marking, the set of marked places is a subset of $P$. The set of all token-to-place assignments is a power set $\Omega = 2^P$ and it is the union of all possible markings. From the property shown in Lemma \ref{lemma:step}, we define an operator called step consequence $\mathcal{T}_F: \Omega \to \Omega$.

\begin{definition} \textbf{(Step consequence operator $\mathcal{T}_F$)} Given an OEN = (P, T, F), $\mathcal{T}_F$ is an operator over the power set $\Omega$:
    \begin{flalign}
        \forall S \in \Omega, \mathcal{T}_F(S) = \{ p_j \in P \, | \, p_j \in S' \gets \land_{p_i\in {^\bullet}t} p_i \in S &\nonumber \\ 
        \text{ where} \, \{p_1, ..., p_n \} \subseteq S \text{ and } p_j \in t^{\bullet}, S' \in \Omega \}& \nonumber
    \end{flalign}
    \label{def:step_operator}
\end{definition} 

Definition 3 says that the result of applying $\mathcal{T}_F$ to a marking is a set of places reachable from it. $\mathcal{T}_F(\emptyset)$ only includes the initial marking due to the antecedent of the implication being false. To obtain the union of all places reachable from an empty marking, we can apply $\mathcal{T}_F$ repeatedly to its results. The recursive application of $\mathcal{T}_F$ can be defined as $\mathcal{T}_F{\uparrow^0} = \emptyset$ and $\mathcal{T}_F{\uparrow^n} = \mathcal{T}_F(\mathcal{T}_F{\uparrow^{n-1}})$. This corresponds to a lattice.

\begin{proposition}
    $(\mathcal{T}_F, \subseteq)$ is a complete lattice.
    \label{proposition:lattice}
\end{proposition} 
\begin{proof}
    The proof follows from Definition \ref{def:step_operator}. $(\mathcal{T}_F, \subseteq)$ is a complete lattice since $\Omega$ is a set partially ordered by set inclusion $\subseteq$ where $\mathcal{T}_F$ is an operator over the power set $\Omega$. The top element $P$ is the union of elements in $\Omega$ and the bottom element $\emptyset$ is the intersection of elements in $\Omega$. 
\end{proof}

Let $(P,T,F)$ be an arbitrary OEN. Using Proposition \ref{proposition:lattice}, we show $\mathcal{T}_F$ over the lattice $(\mathcal{T}_F, \subseteq)$ has a fixpoint. We refer to Tarski's fixpoint theorem \cite{tarski_fixpoint}, which states a monotonic operator on a complete lattice has a least fixpoint. We prove that $\mathcal{T}_F$ is a monotonic operator.

\begin{lemma}
    $\mathcal{T}_F$ is a monotonic operator over $\Omega$. For all $S_1, S_2 \in \Omega$, if $S_1 \subseteq S_2$ then $\mathcal{T}_F(S_1) \subseteq \mathcal{T}_F(S_2)$.
    \label{lemma:monotonicity}
\end{lemma}
\begin{proof}
    This follows from Definition \ref{def:step_operator}. We will always extend a marking by adding places reachable from it. Places in the current marking are reachable and are retained after applying $\mathcal{T}_F$. Therefore, $\mathcal{T}_F$ is a monotonic operator.
\end{proof}

We show that the least fixpoint of $\mathcal{T}_F$ is the union of reachable places by proving that $\mathcal{T}_F$ is a compact operator.

\begin{lemma}
    $\mathcal{T}_F$ is a compact operator over $\Omega$. For all $S_1 \in \Omega$, if $p \in \mathcal{T}_F(S_1)$ then $p \in \mathcal{T}_F(S_2)$ for a finite $S_2 \subseteq S_1$.
    \label{lemma:compactness}
\end{lemma}
\begin{proof}
    Assume $p \in \mathcal{T}_F(S_1)$. From Definition \ref{def:step_operator} and the assumption, since $\mathcal{T}_F$ performs a step, a transition $t$ exists and $S_1$ must contain all pre-set places. Take places $S_2$ that contains $S_1$. From Lemma 1, we know that firing $t$ give $S' \subseteq \mathcal{T}_F(S_2)$ and $p \in S'$.
\end{proof}

\begin{proposition}
    The operator $\mathcal{T}_F$ has a least fixpoint equal to $\mathcal{T}_F{\uparrow^\alpha} = \cup_{n \geq 0}\mathcal{T}_F{\uparrow^n}$.
    \label{proposition:fixpoint}
\end{proposition} 
\begin{proof}
    $\mathcal{T}_F$ is a monotonic function. According to Tarski's fixpoint theorem \cite{tarski_fixpoint}, the function $\mathcal{T}_F$ has a least fixpoint $\mathcal{T}_F{\uparrow^\alpha}$ for some $\alpha \geq 0$. Since $\mathcal{T}_F$ is also a compact operator, if $p \in \mathcal{T}_F(\mathcal{T}_F{\uparrow^\alpha})$ then $p \in \mathcal{T}_F(\mathcal{T}_F{\uparrow^n})$ for some $n \geq 0$ and  $\mathcal{T}_F{\uparrow^n} \subseteq \mathcal{T}_F{\uparrow^\alpha}$. The union of subsets of $\mathcal{T}_F{\uparrow^\alpha}$ is $\cup_{n \geq 0}\mathcal{T}_F{\uparrow^n}$, which covers every member of the least fixpoint. Therefore, $\mathcal{T}_F{\uparrow^\alpha} = \cup_{n \geq 0}\mathcal{T}_F{\uparrow^n}$.
\end{proof}

Proposition \ref{proposition:fixpoint} shows that when we repeatedly apply the step consequence operator to the empty set, the least fixpoint is the union of reachable places. This allows us to compute this least fixpoint by iteratively computing steps and expanding the set of reachable places. We then show the correctness of finding reachable places in an OEN via evaluating the transformed linear and immediately recursive program

\begin{theorem}
    Let $\mathcal{P}$ be a datalog program transformed from an OEN. For constant symbols $c_1$, $c_2$ in $\mathcal{P}$, $\mathcal{P} \models r(c_1,c_2)$ if and only if $c_2$ is reachable from $c_1$ in the OEN given an initial marking $m$.
    \label{theorem:appendix_equivalence}
\end{theorem}
\begin{proof}
    Prove the forward implication. Assume $\mathcal{P} \models r(c_1,c_2)$. Constant symbols $c_1$ and $c_2$ represent place nodes in the OEN. We know $r(c_1,c_2)$ is in the least model. Based on the definition of the Immediate Consequence Operator \cite{van_emden_semantics_1976}, $r(c_1,c_2) \gets \, B_1,..., B_i$ ($i \geq 0$) is a clause in $\mathcal{P}$ and $\{B_1,..., B_i \} \subseteq \mathcal{T_P}{\uparrow^\omega}$. From Definition \ref{def:transformation}, we know that $B_1,..., B_i$ are unit ground clauses describing the pre-sets and post-sets of transitions. There exists $n \geq 0$ such that $c_1 \in \mathcal{T}_F{\uparrow^n}$. Following the transitions $B_1,..., B_i$, we can apply the operator $\mathcal{T}_F$ to $\mathcal{T}_F{\uparrow^n}$. From Definition \ref{def:step_operator} and proposition \ref{proposition:fixpoint}, we know $c_2 \in \mathcal{T}_F{\uparrow^\alpha}$ so a set of marked places $S \in \Omega$ must contain $c_2$. Therefore, $c_2$ is reachable from $c_1$ given an initial marking $m$.  
    
    Now we prove the backward implication. Assume $c_2$ is reachable from $c_1$ in the OEN given an initial marking $m$. We know from Proposition \ref{proposition:fixpoint} that $c_1, c_2\in \mathcal{T}_F{\uparrow^\alpha}$. From Definition \ref{def:step_operator}, there exists a finite sequence of steps $m \xrightarrow[]{t_1} m_1 \xrightarrow[]{t_2} ... \xrightarrow[]{t_{k}} m'$ where $m(c_1) = m'(c_2) = 1$ and $k \geq 0$. From Definition \ref{def:transformation}, the transformed program $\mathcal{P}$ contains these transitions as unit ground clauses $\{B_1,..., B_k\} \subseteq \mathcal{B_P}$  and recursive clauses $r(X,Y) \gets r_1(X,Y)$ and $r(X,Y) \gets r_1(X,Z), r(Z,Y)$. There is a grounding of the recursive clauses such that $r(c_1,c_2) \gets B_1,..., B_k$. By applying the $\mathcal{T_P}$ operator, we know $\mathcal{P} \models r(c_1,c_2)$. 
\end{proof}
Theorem \ref{theorem:appendix_equivalence} says that a program transformed from an OEN is logically equivalent to the OEN based on the fixpoint semantics in Proposition \ref{proposition:fixpoint} for evaluating OEN reachability queries. In the next section, we explain how we leverage Boolean matrix operations to evaluate datalog transformed from OENs.

\section{Boolean matrix encoding}
\label{app:matrix_creation}
From a datalog program $\mathcal{P}$, we define mappings between constant symbols and binary codes. Consider a ground term $u$ that concatenates $n$ constant symbols $c_1, c_2, ..., c_n$. We can represent $u$ using a $n$-bit binary code. When $c_k$ is in $u$ the $k$-th bit of the binary code is 1, otherwise, it is 0. For a linear and immediately recursive clause $r$ that takes such terms as arguments, two Boolean matrices $\textbf{R}_1$ and $\textbf{R}_2$ are defined by mapping unit ground clauses used by $r$ to binary codes. The two matrices in logic program format have $\{v_1(1, b_1^1)$, $v_1(2, b_1^2), ...\}$ and $\{v_2(1, b_2^1)$, $v_2(2, b_2^2), ...\}$ as their rows where $v_1(i, b_1^i)$ and $v_2(i, b_2^i)$ are the $i$-th rows of the Boolean matrices $\textbf{R}_1$ and $\textbf{R}_2$. In $v_1$ and $v_2$ clauses, $b_1^i$  and $b_2^i$ are binary codes to represent constant symbols $u_1$ and $u_2$ as $r$'s first and second arguments. $\textbf{R}_1$ and $\textbf{R}_2$ are defined such that $v_1(i, b_1^i)$ and $v_2(i, b_2^i)$ hold only when $\mathcal{P} \models r(u_1, u_2)$. 
\begin{example}
\textit{We show Boolean matrices created from $\mathcal{P}_2$ in Example \ref{ex:transformation}. Individual metabolites are represented by constant symbols $c1$, $c2$ and $c3$. Sets of chemical substances are denoted by additional concatenated constant symbols, $c1\_c2$, $c3\_c4$. On the right-hand side, we create two Boolean matrices $\textbf{R}_1$ and $\textbf{R}_2$ from the reaction clauses in $\mathcal{P}_1$. Since the first reaction\_1 clause also contains the constant symbols of the other reaction\_1 clauses, to avoid redundancy, only one binary code is created for these clauses in $\textbf{R}_1$ and $\textbf{R}_2$. }\\
\begin{minipage}{0.57\textwidth}
\begin{gather}
\mathcal{P}_{2}:\left\{
\begin{aligned}
    ... &\\
    reaction\_1(c1\_c2, c3\_c4). &\\
    reaction\_1(c1\_c2, c3). &\\
    reaction\_1(c1\_c2, c4).&\\
    reaction\_2(c3\_c4, c5). &\\
    reaction\_3(c5, c4). &\\
    ... &\\
    reaction(X,Y) \gets reaction\_1(X,Y). &\\
    reaction(X,Y) \gets reaction\_2(X,Y). &\\
    reaction(X,Y) \gets reaction\_3(X,Y). &\\
    pathway(X,Y) \gets reaction(X,Y). &\\ 
    pathway(X,Y) \gets reaction(X,Z), &\\
                pathway(Z,Y). & \nonumber
\end{aligned}
\right\}
\label{ex:recursive_program}
\end{gather}
\end{minipage}
\begin{minipage}{0.4\textwidth}
\vspace{-5pt}
\begin{gather}
\textbf{R}_1 \text{ (matrix)}: 
\begin{bmatrix}
    0 \quad 0 \quad 0 \quad 1 \quad 1 \\
    0 \quad 1 \quad 1 \quad 0 \quad 0 \\
    1 \quad 0 \quad 0 \quad 0 \quad 0 \nonumber
\end{bmatrix}
\end{gather}
\vspace{-10pt}
\begin{gather}
\textbf{R}_2 \text{ (matrix)}: 
\begin{bmatrix}
    0 \quad 1 \quad 1 \quad 0 \quad 0 \\
    1 \quad 0 \quad 0 \quad 0 \quad 0 \\
    0 \quad 1 \quad 0 \quad 0 \quad 0 \nonumber
\end{bmatrix}
\end{gather}
\vspace{-10pt}
\begin{gather}
\textbf{R}_1\text{ (program)}: \left\{
\begin{aligned}
    &v_1(1, 00011). \\ 
    &v_1(2, 01100). \\
    &v_1(3, 10000). \nonumber
\end{aligned}
\right\}
\end{gather}
\vspace{-10pt}
\begin{gather}
\textbf{R}_2\text{ (program)}: \left\{
\begin{aligned}
    &v_2(1, 01100). \\
    &v_2(2, 10000). \\
    &v_2(3, 01000). \nonumber
\end{aligned}
\right\}
\end{gather}
\end{minipage}
\bigskip

\textit{Constant symbols are represented by rows in $\textbf{R}_1$ and $\textbf{R}_2$. Constant symbols $c1\_c2$, $c3\_c4$ and $c5$ that appear on the left-hand side of reactions are denoted by binary bits in $\textbf{R}_1$. Constant symbols $c3$, $c4$, $c5$ and $c3\_c4$ that appear on the right-hand side of reactions are denoted by binary bits in $\textbf{R}_2$. Each row in the Boolean matrices is written as a unit ground clause. For example, since $\mathcal{P}_2 \models reaction(c1\_c2, c3\_c4)$, the first row of $\textbf{R}_1$ is written a binary code ``00011'' to represent $c1\_c2$ and the first row of $\textbf{R}_2$ is a binary code ``01100'' to represent $c3\_c4$. Each ``1'' binary bit indicates the presence of a constant symbol. }

\label{ex:bmlp_compilation}
\end{example}

Each transformed program contains clauses describing an initial marking. These clauses are represented by a Boolean vector $\textbf{v}$ which is a single-row Boolean matrix where the binary bits correspond to the constant symbols.

\begin{example}
\textit{We show the mapping of a single row matrix from the token definition from $\mathcal{P}_2$ in Example \ref{ex:transformation}. Individual metabolites are represented by constant symbols $c1$ and $c2$. The set of chemical substances is denoted by an additional concatenated constant symbol, $c1\_c2$.}\\
\begin{minipage}{0.57\textwidth}
\begin{gather}
\mathcal{P}_{1}:\left\{
\begin{aligned}
    token(m1,c1\_c2). &\\
    token(m1,c1).&\\
    token(m1,c2).&\\
    ...&\\
    reaction(X,Y) \gets token(X,Y).&\\
    ... &\nonumber
\end{aligned}
\right\}
\end{gather}
\end{minipage}
\begin{minipage}{0.4\textwidth}
\vspace{-5pt}
\begin{gather}
\textbf{v} \text{ (vector)}: 
\begin{bmatrix}
    0 \quad 0 \quad 0 \quad 1 \quad 1 \nonumber
\end{bmatrix}
\end{gather}
\vspace{-10pt}
\begin{gather}
\textbf{v} \text{ (program)}: \left\{
\begin{aligned}
    v(1, 00011). \nonumber
\end{aligned}
\right\}
\end{gather}
\end{minipage}
\bigskip

\textit{On the right-hand side, the constant symbol $c1\_c2$ is represented by the vector $\textbf{v}$. Since the first token clause contains constant symbols in the other token clauses, $\textbf{v}$ only includes the right-hand side of the first token clause to avoid redundancy. In program form, $\textbf{v}$ is written as a unit ground clause containing a binary code ``00011'' to represent $c1\_c2$. }
\end{example}

\section{Iterative extension (BMLP-IE)}
\label{app:bmlp_ie}
\begin{algorithm}[t]
    \caption{Iterative extension (BMLP-IE)}
    \label{alg:algorithm1}
    \textbf{Input}: A $1 \times n$ vector $\textbf{v}$, a $1 \times k$ vector $\textbf{t}$ that represents transitions, two $k \times n$ Boolean matrices $\textbf{R}_{1}$, $\textbf{R}_{2}$ that encode unit ground clauses.\\
    \textbf{Output}: Transitive closure $\textbf{v}^*.$
    \begin{algorithmic}[1] 
        \STATE Let $\textbf{v}^*$ = $\textbf{v}$, $1 \times k$ vector $\textbf{v}'$ = \textbf{0}.
        \WHILE{$True$}
            \FOR{$1 \leq i \leq k$} 
                \STATE $(\textbf{v}')_i = (\textbf{t})_i$ if ($\textbf{v}$ AND $(\textbf{R}_{1})_{i,*}$) == $(\textbf{R}_{1})_{i,*}$. 
            \ENDFOR
        \STATE $\textbf{v}^*$ = ($\textbf{v}'$ $\textbf{R}_{2}$) + \textbf{v}.
        \STATE $\textbf{v} = \textbf{v}^*$ if $\textbf{v}^* \ne \textbf{v}$ else break.
        \ENDWHILE
    \end{algorithmic}
\end{algorithm}
The iterative extension algorithm (BMLP-IE) operates on two matrices to iteratively expand terms reachable in a partially grounded query $r(u, Y)$. Given that $\mathcal{P}$ contains $n$ constant symbols, we represent the constant symbols in $u$ as a $1 \times n$ row vector or a $n$-bit binary code $\mathbf{v}$. We write $k \times n$ Boolean matrices $\textbf{R}_1$ and $\textbf{R}_2$ for $k$ unit ground clauses used by $r$. $\textbf{R}_1$ and $\textbf{R}_2$ describe which constant symbols appear first and second in the unit ground clauses. In addition, a $1 \times k$ Boolean vector $\textbf{t}$ is used to express additional control on the viability of transitions in the OEN. Transition can be turned ``on'' or ``off'' via this vector. ``On'' transitions have the binary value 1 and ``off'' transitions have the value 0. This allows us to temporarily modify the network connectivity without re-creating $\textbf{R}_1$ and $\textbf{R}_2$.

We implement BMLP-IE (Algorithm \ref{alg:algorithm1}) in SWI-Prolog (version 9.2) to compute the transitive closure of $\textbf{v}$. In SWI-Prolog, each binary code can be written as an integer number, e.g. ``00011'' is the integer ``3''. We use the bitwise AND to locate rows in $\textbf{R}_1$ that contain all the 1 elements in $\textbf{v}$. This computes indices of unit ground clauses that can extend the set of constant symbols in $\textbf{v}$. These constants are added to $\textbf{v}$ via Boolean matrix multiplication and addition. We iterate this process until we find the elements of $\textbf{v}$ no longer change. 
\begin{theorem}
    Given a $1 \times n$ vector $\textbf{v}$, a $1 \times k$ vector $\textbf{t}$, two $k \times n$ Boolean matrices $\textbf{R}_{1}$ and $\textbf{R}_{2}$, Algorithm \ref{alg:algorithm1} has a time complexity $O(k n^2)$ for computing the transitive closure $\textbf{v}^*$.
    \label{theorem:2}
\end{theorem}
\begin{proof}
    In Algorithm \ref{alg:algorithm1}, the complexity of an iteration in the ``while'' loop is $O(k \times n)$ bitwise operations due to the multiplication between a vector and a $k \times n$ matrix. Each ``for'' loop pass also performs $O(k \times n)$ bitwise comparisons. Until we find the transitive closure, at least one unit ground clause needs to be added in each ``while'' loop iteration. Therefore, there are at most $n$ iterations which require $O(k \times n^2)$ bitwise operations. $\square$
\end{proof}

Theorem \ref{theorem:2} shows that BMLP-IE has a polynomial runtime cubic in the size of the Boolean matrices. Evaluating the program $\mathcal{P}$ using SLD resolutions has the same time complexity, and the runtime of BMLP-IE primarily benefits from bit-wise operations. For SLD resolutions, the size of the subset of the Herbrand base provable from a query $r(u, Y)$ is bounded by $O(n)$. Each member of this subset can be found via a depth-first search. Since there are $k$ unit ground clauses and each variable can be grounded by at most $n$ constants, this query requires at most $n$ search in a tree of $O(k \times n)$ nodes. For OEN reachability queries, it is important to note that the size of Boolean matrices grows linearly with the number of nodes because they are both bounded by the number of transitions. 

\section{Active learning sample complexity}
\label{app:sample_complexity}
Theorem \ref{theorem:active_learning} shows the sample complexity gain of active learning over learning from randomly sampled instances. 
\begin{theorem}[Active learning sample complexity bound]
    For some $\phi \in [0, \frac{1}{2}]$ and small hypothesis error $\epsilon > 0$, if an active version space learner can select instances to label from an instance space $\mathcal{X}$ with minimal reduction ratios greater than or equal to $\phi$, the sample complexity $s_{active}$ of the active learner is 
    \begin{flalign}
        s_{active} \leq \frac{\epsilon}{\epsilon + \phi} s_{passive} + c
    \end{flalign}
    where c is a constant and $s_{passive}$ is the sample complexity of learning from randomly labelled instances.
\end{theorem}

\begin{proof}
    We consider an arbitrary distribution of the minimal reduction ratio over $\mathcal{X}$ bounded by $[0,\frac{1}{2}]$. Suppose a version space learner can select instances with a minimal reduction ratio greater than $\phi$. Since $\phi$ represents the ratio of the minority partition, the number of candidate hypotheses remaining in the version space $V_s$ after $s$ labelled instances is bounded by the size of the majority partitions $(1 - \phi)^s |H|$ where $H$ is the hypothesis space. The probability that any hypothesis in $V_s$ with an error larger than $\epsilon$ is consistent with $s$ training examples is: 
    \begin{flalign*}
        p_{error} &= (1 - \epsilon)^s |V_s| \\ \nonumber
        &\leq (1 - \epsilon)^s (1 - \phi)^s |H|
    \end{flalign*}
    $p_{error}$ should be small so suppose for small $\delta > 0$,
    \begin{flalign}
        (1 - \epsilon)^s (1 - \phi)^s |H| = \delta
        \label{eq:p_error}
    \end{flalign}
    Equation \ref{eq:p_error} corresponds to the probability of the extreme case when only the minority of the remaining hypotheses can be eliminated by each labelled instance. By rearranging Equation (\ref{eq:p_error}), the number of training examples required by this learner in this extreme case is:
    \begin{flalign*}
        s = \frac{1}{- ln(1 - \epsilon) - ln(1 - \phi)} (ln(|H|) + ln(\frac{1}{\delta}))
    \end{flalign*}
    For an arbitrary active learner whose membership queried instances also have a minimal reduction ratio greater than or equal to $\phi$, the number of actively selected instances $s_{active}$ is upper bounded by $s$ for reducing the size of version space to $|V_s|$. Since $(1 - u) < e^{-u}$ for small $u > 0$, it is true that $- ln(1 - u) > u$. So the following holds:
    \begin{flalign*}
        s_{active} < \frac{1}{\epsilon + \phi} (ln(|H|) + ln(\frac{1}{\delta}))
    \end{flalign*}    
    According to the Blumer bound \cite{blumer_occams_1987}, for a version space learner that randomly samples examples to produce a hypothesis of error at most $\epsilon$ with a probability less than or equal to $\delta$, the sample complexity is $s_{passive} > \frac{1}{\epsilon} (ln(|H|) + ln(\frac{1}{\delta}))$. Following this, it must be true that $ \epsilon \, s_{passive} > (\epsilon + \phi) \, s_{active}$. So we can derive the Inequality \ref{eq:sample_complexity_bound} in the theorem:
    \begin{flalign*}
        s_{active} \leq \frac{\epsilon}{\epsilon + \phi} s_{passive} + c
    \end{flalign*}
    where c is a constant. $ \square $
\end{proof}

\section{Implementation details}
\label{app:implementation}
\subsection{Encoding a metabolic network} 

Our encoding of a metabolic network is based on the Robot Scientist \cite{King04:RobotScientist,ase_progol}. The Robot Scientist is implemented using the ILP system Progol \cite{progol} and finds the producible set of metabolites by traversing a metabolic network. It uses a metabolic network as a deductive database and answers a query of user-defined gene deletion and experimental conditions based on SLD resolutions. Its deductive database is written in Prolog syntax and uses Prolog list operations to manipulate constant symbols representing metabolic substances. Metabolic reactions are enumerated based on the set of metabolites currently producible. The products of a metabolic reaction extend the list of producible metabolites. Since the number of metabolic reactions is finite, the closure of members in this list needs to contain essential metabolites for the growth of an organism. In an auxotrophic mutant experiment outcome prediction, a gene function is removed. Some key metabolic reaction pathways might no longer be viable due to the removal, leading to insufficient metabolites for growth. In this case, the predicted experiment's outcome is a lack of growth and a phenotypic effect. 

In the Robot Scientist, the metabolic network is defined by $enzyme/6$ clauses. Each $enzyme/6$ clause contains an artificial identifier for an enzyme, gene locus identifiers, reaction identifiers, a reaction reversibility indicator and reactant and product identifiers. A gene locus corresponds to a physical site within a genome and is uniquely associated with a gene identifier. An artificial enzyme identifier is created for each unique gene locus identifier. An auxotrophic mutant experiment is described as a query of $phenotypic\_effect/2$ whose outcome is a binary prediction. A $phenotypic\_effect/2$ query has a gene identifier and a list of reagents which are optional chemical substances that provide additional nutrients to facilitate metabolic reactions. These reagents are added to a list of pre-defined base growth medium compounds. Gene deletions are realised via $codes/2$ clauses where an enzyme corresponding to the deleted gene would be removed from the set of metabolic reactions to be enumerated. 
\begin{example}
\textit{In the following program $\mathcal{P}_3$, the $codes/2$ clause says that the gene locus ``b2379'' is responsible for the enzyme ``enz201''. The $enzyme/6$ clause defines the reversible reaction ``alata\_l'' (L-alanine transaminase). This reaction involves reactants ``akg\_c'' (Oxoglutarate), ``ala\_\_L\_c'' (L-alanine) and products ``glu\_\_L\_c'' (L-glutamate), ``pyr\_c'' (Pyruvate). An integer value in \{1,2\} indicates the reaction reversibility and the integer ``2'' says that the reaction is reversible.}
\begin{gather}
\mathcal{P}_3:\left\{
\begin{aligned} 
    codes(b2379,enz201). &\\
    enzyme(enz201,[b2379],[alata\_l],2,[[akg\_c,ala\_\_L\_c]],&\\
    [[glu\_\_L\_c,pyr\_c]]). &\nonumber
\end{aligned}
\right\}
\end{gather}
\textit{The following $phenotypic\_effect/2$ query describes an auxotrophic mutant experiment where the gene ``b2379'' is removed and an additional metabolite ``glyc\_e'' (Glycerol) is introduced in the growth medium. }
\begin{flalign*}
?- phenotypic\_effect(b2379,[glyc\_e]).
\end{flalign*}
\textit{iML1515 separates metabolites into different compartments in the cell. The suffices ``e'' (extracellular space) and ``c'' (cytosol) represent two compartments, which divide the metabolic network into sub-networks.}
\end{example}

We included this encoding in $BMLP_{active}$ for hypothesis generation and human interpretation. The $codes/2$ and $enzyme/6$ clauses enable us to describe the gene-reaction-metabolite mappings and standard identifiers in iML1515\footnote{The iML1515 model is described in the text-based JavaScript Object Notation (JSON). It uses standard identifiers from the Biochemical Genetic and Genomic (BiGG) \cite{king_bigg_2016} knowledge base of large-scale metabolic reconstructions.} into a logic program containing $codes/2$ and $enzyme/6$ clauses. iML1515 uses standard nomenclature so the logic program encoding can be applied to other genome-scale metabolic networks. The logic program encoding can also be linked to publications and external databases to obtain further details on genes, proteins, reactions and metabolites. 

\subsection{BMLP-IE predictions from GEM}

\begin{algorithm}[h]
    \caption{BMLP phenotype prediction}
    \label{alg:GEM_BMLP_IE}
    \textbf{Input}: Background knowledge $BK$, an experiment query $q$.\\
    \textbf{Output}: A binary phenotypic effect classification: 1 represents a phenotypic effect and 0 denotes no phenotypic effect.
    \begin{algorithmic}[1]
        \STATE From $reaction/2$ clauses in $BK$, create Boolean matrices $\textbf{R}_1$, $\textbf{R}_2$.
        \STATE From biomass metabolites in $BK$, create vector $\textbf{v}_{biomass}$.
        \STATE From growth medium metabolites in $q$, create vector $\textbf{v}$.
        \STATE From knockout genes in $q$, create vector $\textbf{t}$ for viable reactions.
        \STATE return 1 if (BMLP-IE($\textbf{v}$, $\textbf{t}$, $\textbf{R}_1$, $\textbf{R}_2$) AND $\textbf{v}_{biomass}) \ne \textbf{v}_{biomass}$ else return 0.
    \end{algorithmic}
\end{algorithm}

To predict phenotypes from iML1515 using BMLP-IE, $reaction/2$ clauses are created from $enzyme/6$ clauses. Each clause contains reactants and product identifiers. Each reversible reaction is defined as two clauses where reactants and products are swapped. $BMLP_{active}$ allows multiple gene deletions, e.g. $phenotypic\_effect([b2290,b2379]$, $[glyc\_e])$. This enables the learning of digenic interactions such as gene-isoenzyme functions. Since multiple genes can be responsible for the same enzyme, single knockouts may not have phenotypic effects. A hypothesis of a gene-isoenzyme function contains one $codes/2$ clause and one $enzyme/6$ clause. The $enzyme/6$ in the hypothesis would be a replicate of an existing clause in the background knowledge with a new enzyme identifier.  

Algorithm \ref{alg:GEM_BMLP_IE} shows the application of BMLP-IE (Algorithm \ref{alg:algorithm1}) for finding the transitive closure of producible metabolites. The two Boolean matrices inputs to the BMLP-IE algorithm represent reactants and products participating in the reactions. The growth medium condition is encoded as the input vector to the BMLP-IE algorithm. This vector would be extended to represent metabolites that are producible by the organism in an auxotrophic mutant experiment. Specifying knockout genes excludes a subset of reactions. A vector computed during an intermediate step in the BMLP-IE algorithm denotes viable reactions given the set of knockout genes.  A viable reaction is indicated by a binary bit 1 in this vector. Blocked reactions are excluded based on $codes/2$ and $enzyme/6$ definitions. In BMLP-IE, Boolean operations compute the transition between metabolic states and return the closure set of producible metabolites. A binary phenotypic classification is made according to the inclusion of essential metabolites in the closure set, indicating cell growth comparable to an unedited wild-type MG1655 \textit{E. coli} strain. A binary classification of ``1'' corresponds to a phenotypic effect and ``0'' is no phenotypic effect. The essential metabolites come from the biomass function of the wild-type in iML1515 \cite{iML1515}\footnote{We use the wild-type biomass function specified by iML1515 from http://bigg.ucsd.edu/ \cite{king_bigg_2016}. We use the qualitative version of this biomass function, meaning that we only care about the metabolites in the biomass, not their reaction coefficients.}. The biomass function describes the necessary metabolite composition for the growth of the organism. 

\section{Experimental details}
\label{app:experiments}

Since a logic program representing a GEM describes many biological entities and intricate relations between them, abduction of gene function hypotheses and in-vivo validations requires exploring a large hypothesis space and instance space. Both are challenging computationally and empirically. $BMLP_{active}$ addresses these by evaluating a large logic program with Boolean matrices and actively selecting experiment instances to perform. We demonstrated that the $BMLP_{active}$ framework can efficiently perform phenotypic predictions and actively suggest informative experiment instances that can reduce experiment costs. 

\subsection{Experimental cost}

Here, we defined the experimental cost function as the cost of the reagents, which are optional nutrient substances used in selected experiments. We also assumed that when added to growth media, reagents are used in the same volume. In addition, we did not explicitly account for plastic consumables, equipment deprecation or specify the time to carry out each experiment, which we consider fixed and are affected mainly by the number of experiments. We also did not account for costs incurred due to different experiment types, such as different measurement timescales that prevent experiments from being performed concurrently. To maintain consistent scaling for comparisons, the total experimental reagent cost was the amount of unit cost required by selected experiments. We considered the unit cost of the reagents as the cost of the cheapest reagent.  We provided the costs of reagents in Table \ref{table:optional_nutrient_cost1} and \ref{table:optional_nutrient_cost2}. We calculated the total experimental reagent cost as the sum of the reagent costs in selected experiments divided by the unit cost.\\

\subsection{Experiment 1: BMLP-IE predictions}
\label{app:experiment1}

\noindent \textit{Materials.} We used $BMLP_{active}$ to predict a phenotype in each auxotrophic single-knockout mutant experiment. Every experiment removes one gene from the biological system and uses a growth medium. The GEM iML1515 has been validated against 16 carbon sources as optional nutrients in \cite{iML1515}. We focused on the same experimental conditions for base SWI-Prolog and $BMLP_{active}$. Experimental phenotypic data on the same 16 carbon sources and FBA phenotypic predictions of iML1515 are already processed as binary classes in \cite{iML1515}. In \cite{iML1515}, measured and predicted mutant growth rates at a proportion of the wild-type's growth rate greater than a fixed threshold are mapped to the no phenotypic effect class. \\

\noindent \textit{Methods.} We compared the runtime of $BMLP_{active}$ and the Robot Scientist \cite{King04:RobotScientist,ase_progol}'s encoding of the metabolic network (Appendix \ref{app:implementation}). We first compared the mean and standard deviation in runtime of our new logical matrix encoding against the SWI-Prolog encoding. In addition, we explored multi-threading with both encodings. We used the multi-threading library offered by SWI-Prolog. Owing to the overhead of managing threads, multi-threading Boolean operations with SWI-Prolog have diminishing returns. In BMLP-IE, conditions are encoded as matrices, enabling us to predict in batches. For BMLP-IE, we used the concurrent threading library in SWI-Prolog at the batch level. The runtime experiments are repeated 10 times.

\subsection{Experiment 2: active learning sample complexity and experiment cost}
\label{app:experiment2}

\noindent \textit{Materials.}  Up to 3 of 16 common carbon sources were used as additional nutrients to the base growth medium. Each of these 16 carbon sources is an optional nutrient in the experimental data for validating \cite{iML1515}. This set contains preferred carbon sources \cite{martinez-gomez_new_2012,zampieri_regulatory_2019} for \textit{E. coli} such as glucose and glycerol. We used iML1515 to predict with up to 3 optional carbon sources, generating $\sum_{0}^{i=3}\binom{16}{i} = 697$ synthetic data. This allows us to go beyond the 16 experimental data points for each gene knockout in \cite{iML1515}. We observed a single-knockout phenotypic effect from 213 genes. We treated these synthetic data as ground truths, where 213 genes have only positive phenotype examples. We only focused on a subset of 3 genes because they have both positive and negative examples\footnote{The distribution of training examples for this experiment is described in Table \ref{table:example_distribution} of Appendix \ref{app:appendix_distribution}.}. Then, we removed gene functions associated with these genes from the background knowledge. Thus, the resulting incomplete model has a smaller set of $codes/2$ clauses. \\

\noindent \textit{Methods.} We constrained the selection of experiments by an experimental reagent cost budget. After experiment selection, both methods returned the hypothesis with the single highest compression, otherwise, a hypothesis was selected from those sharing the highest compression randomly. We focused on one removed gene function at a time to recover it. Each hypothesis, except the empty one, is associated with a gene locus identifier. We concentrate on 1295 candidate hypotheses that represent the non-maintenance reactions in iML1515 (including an empty hypothesis). The instance space had experiment candidates with optional growth nutrients up to 3 carbon sources, which was a total of 697 synthetic data points. 

\subsection{Experiment 3: learning digenic functions}
\label{app:experiment3}

\noindent \textit{Materials.} We focused on re-discovering a target isoenzyme function associated with the gene $tyrB$ in the Tryptophan biosynthesis pathways. We aim to demonstrate gene-isoenzyme function learning by recovering this gene function. Associated functions of $tyrB$ and $aspC$ can be classified as isoenzyme functions since they are responsible for producing an important amino acid L-phenylalanine. In our logical representation of iML1515, we use reaction and $codes/2$ unit ground literals to describe an isoenzyme. We selected 33 genes associated with this essential pathway. In this focus gene set, 27 genes are associated with a single function and 6 genes are related to two functions. We consider medium conditions with 7 optional nutrients, including aromatic amino acids for recovering the deleted $tyrB$ isoenzyme function. Aromatic amino acids are fundamental for protein synthesis but also serve as precursors for vital secondary metabolites and high-value chemicals \cite{choi_violacein_2015}. \\

\noindent \textit{Methods.} We removed tyrB's gene function and metabolic reaction associated with this isoenzyme from $BMLP_{active}$. The experiment instance space contained double gene-knockout and single gene-knockout experiments of 33 genes. There were $\binom{33}{2}$ gene pair combinations for double-gene knockout experiments and 33 single-gene knockout experiments. We considered growth medium conditions with the 7 optional nutrients. Therefore, the experiment instance space was $(\binom{33}{2} + 33) \times 7 = 3927$ and 3696 experiment instance labels were synthetically generated from double-knockout experiments using the full iML1515 background knowledge. The rest 231 single-knockout empirical data were from validation data in \cite{iML1515}. 

In the metabolic network, an isoenzyme function hypothesis contains a node represented by a new enzyme identifier (a $enzyme/6$ clause) and a new link between a gene locus identifier and this new enzyme identifier (a $codes/2$ clause). The hypothesis space contained associations of 33 genes (from the focused library) with 33 $enzymes/6$ clauses. There were $(27 \times 32 + 6 \times 31) = 1050$ potential hypotheses for 27 single-function genes and 6 double-function genes. We also included the tyrB function hypothesis and an empty hypothesis, which gave 1052 hypotheses in total. We increased linearly and exponentially the number of experiments $N$ that could be selected according to \{1, 2, 4, 5, 8, 10, 15, 16, 20, 25, 30, 32, 64, 128, 256\}. $BMLP_{active}$ selected $N$ experiments from this instance space to actively learn from the hypothesis space. The random selection strategy randomly sampled $N$ instances from the instance space. Again, both methods output the hypothesis with the highest compression according to the sampled examples. A hypothesis was randomly selected when multiple competing hypotheses shared the highest compression. We recovered the models by adding the final hypotheses. We recorded the number of experiments selected and evaluated if the recovered models contained the target isoenzyme function from the iML1515 background knowledge. Each experiment selection method was repeated 10 times, and we computed the frequency of successful isoenzyme recovery across these 10 repeats. \\

\newpage
\section{Reagent costs of optional nutrient substances} 
\label{app:appendix_cost}
\begin{table}[h]
    \centering
	\begin{tabular}{c|c}
	      Carbon Source & Cost per gram (\pounds/g) \\
		\hline
		D-Mannitol & 0.534\\
            \hline
		Xylose & 0.51\\
            \hline
		Galactose & 0.0936\\ 
            \hline
            Pyruvate & 0.2776\\ 
            \hline
            Ribose & 0.993\\ 
            \hline
            Oxoglutarate & 0.437\\ 
            \hline
            Acetate & 0.0748\\ 
            \hline
            Gluconate & 0.0359\\ 
            \hline
            Glycerol & 0.07\\ 
            \hline
            Maltose & 0.244\\ 
            \hline
            Lactate & 0.0546\\ 
            \hline
            Sorbitol & 0.196\\ 
            \hline
            Succinate & 0.0825\\ 
            \hline
            N-acetyl-glucosamine & 2.93\\ 
            \hline
            L-Alanine & 0.81\\ 
            \hline
            Glucose & 0.0476\\ 
            \hline
	\end{tabular}
\caption{The costs of 16 optional carbon sources considered in Experiment 2.}
\label{table:optional_nutrient_cost1}
\end{table}

\begin{table}[h]
    \centering
	\begin{tabular}{c|c}
	      Carbon Source & Cost per gram (\pounds/g) \\
		\hline
		Indole & 0.0428\\ 
            \hline
		Phenylalanine & 0.0346\\
            \hline
		Tryptophan & 0.0346\\ 
            \hline
            Tyrosine & 0.0346\\ 
            \hline
            Shikimate & 8.2\\ 
            \hline
            L-Aspartic acid & 0.0156\\ 
            \hline
	\end{tabular}
\caption{The costs of aromatic amino acids considered in Experiment 3.}
\label{table:optional_nutrient_cost2}
\end{table}

\newpage
\section{Gene function learning example distribution in Experiment 2. } 
\label{app:appendix_distribution}
\begin{table}[h]
\centering
\caption{The distribution of synthetic positive and negative examples for each gene knockout based on iML1515. Gene identifiers start with the lower-case letter ``b'' and gene descriptive names are given in parentheses. }
	\begin{tabular}{c|c|c}
	      Gene Id. & $|E^+|$ & $|E^-|$ \\
            & & \\
		\hline
		& & \\
		b0720 (gltA) & 576 & 121\\ 
		& & \\ \hline
            & & \\
		b1136 (icd) & 576 & 121 \\ 
            & & \\ \hline
            & & \\
		b3729 (glmS) & 576 & 121 \\
            & & \\ \hline
            & & \\
            The other 210 genes & 697 & 0 \\
            & & \\
	\end{tabular}
\label{table:example_distribution}
\end{table}

In Experiment 2, for the 16 carbon sources, 213 genes have a single-knockout phenotypic effect.  Three genes, b0720 (gltA), b1136 (icd) and b3729 (glmS), have both positive and negative examples. We focus on b0720 and b3729 since b0720 and b1136 have the same coverage over the 697 training examples so their functions cannot be distinguished.

\end{appendix}

\end{document}